\documentclass{sig-alternate}

\usepackage{epstopdf}

\newtheorem{defi}{Definition}[section]
\newtheorem{theorem}{Theorem}[section]

\newtheorem{prop}[theorem]{Proposition}

\begin{document}

\title{Scalable Multi-Database Privacy-Preserving Record Linkage using Counting Bloom Filters}
%\subtitle{Technical Report}
%\titlenote{A full version of this paper is available as

\author{
\alignauthor Dinusha Vatsalan, Peter Christen, and Erhard Rahm$^\dag$ \\
       \affaddr{Research School of Computer Science,} 
       \affaddr{The Australian National University}\\
       \affaddr{Canberra ACT 0200, Australia}\\
       \affaddr{$^\dag$Universit{$\ddot{\mbox{a}}$}t Leipzig,}
       \affaddr{Institut f{$\ddot{\mbox{u}}$}r Informatik,}
       \affaddr{04109, Leipzig, Germany}\\
       \email{\{dinusha.vatsalan, peter.christen\}@anu.edu.au,$^\dag$rahm@informatik.uni-leipzig.de}
}

\maketitle
\begin{abstract}
Privacy-preserving record linkage (PPRL) aims at integrating sensitive information from multiple disparate databases of different organizations. PPRL approaches are increasingly required in real-world application areas such as healthcare, national security, and business. Previous approaches have mostly focused on linking only two databases as well as the use of a dedicated linkage unit. Scaling PPRL to more databases (multi-party PPRL) is an open challenge since privacy threats as well as the computation and communication costs for record linkage increase significantly with the number of databases. We thus propose the use of a new encoding method of sensitive data based on \emph{Counting Bloom Filters} (CBF) to improve privacy for multi-party PPRL. We also investigate optimizations to reduce communication and computation costs for CBF-based multi-party PPRL with and without the use of a dedicated linkage unit. Empirical evaluations conducted with real datasets show the viability of the proposed approaches and demonstrate their scalability, linkage quality, and privacy protection.

\end{abstract}


\section{Introduction}
\label{sec:intro}

A wide range of real-world applications, including in healthcare,
government services, crime and fraud detection, 
national security, and businesses, require person-related data from multiple sources held by different organizations to be integrated or linked. Integrated data can then be used for data mining and analytics to empower efficient and quality decision making with rich data.  
Integrating data helps improving the quality of data by 
identifying and resolving conflicts in
data values, enriching data with additional detailed information, and
dealing with missing
values~\cite{Chr12}. 

The analysis and mining of data integrated across organizations can be used, for example,
in health outbreak systems that allow the early detection of infectious diseases before they
spread widely around a country or even worldwide. %~\cite{Muru10}.
Such an application requires data to be integrated across several sources, 
including
human health data, travel data, consumed drug data, 
and even animal health data~\cite{Clif04}.
A second contemporary motivating example is
national security applications that integrate data
from law enforcement agencies, Internet service
providers, businesses, as well as financial 
institutions 
to enable the accurate
identification of crime and fraud, or of terrorism suspects~\cite{Phu12}.

In the absence of unique entity identifiers in the databases
that are to be linked, 
personal identifying attributes (such as names and 
addresses) need to be used for the
linkage. Known as quasi-identifiers (QIDs)~\cite{Vat14}, 
such attribute values are in general assumed to be
sufficiently well correlated with
entities to allow accurate linkage. Using such personal information
across different organizations,
however, often leads to privacy and confidentiality concerns. 
This problem has been addressed through the development of
`privacy-preserving record
linkage' (PPRL)~\cite{Vat13} techniques. PPRL aims to
conduct linkage using only masked (encoded) QIDs
without requiring any sensitive or confidential information
to be exchanged and revealed between the organizations
involved in the linkage.
Generally,
masking is conducted on QIDs 
to transform the original values such that a specific
functional relationship exists between the original and the masked
values~\cite{Vat14}.
While there have been many different approaches proposed for
PPRL (as reviewed in~\cite{Vat13}), most work  
thus far has concentrated on
linking records from only two sources (or parties). 
As the healthcare and national security examples
described above show, linking data from several
sources is however commonly required in practical applications.

\begin{figure*}[!t]
  \centering
  \scalebox{0.85}[0.6]{\includegraphics[width=1.0\textwidth]
                      {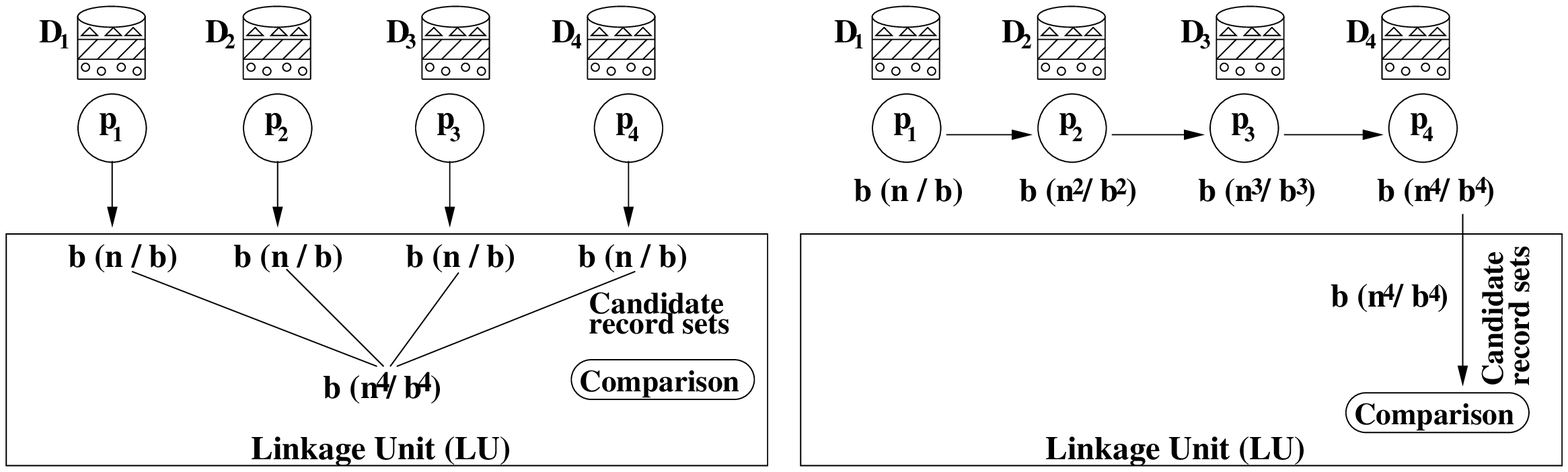}}
  \caption{\small{An overview of traditional 
na\"{i}ve comparison of candidate record sets masked using BFs (left)
and CBFs (right, as will be described in detail in Section~\ref{sec:algorithm}) from $p=4$ parties
using a $LU$.
Databases are indexed/blocked
to reduce the number of candidate record sets such that only records in the same blocks 
are compared and classified. Blocks are illustrated by different patterns in $\mathbf{D}_i$, with $1 \le i \le p$.
$n$ denotes the number of records in the databases (assuming all databases are of equal size)
and $b$ denotes the number of blocks generated (assuming blocks of equal size).
Independent of the masking function and the communication pattern used 
(i.e.\ BFs and direct one-to-one communication in the left figure, and 
CBFs and ring-based communication in the right figure), 
the na\"{i}ve approach
results in exponential complexity of $b~(n^4/b^4)$ candidate sets each consisting of $p=4$
records (one from each party).
}
}
           \label{fig:PPRL_naive}
\end{figure*}

The drawback of the small number of existing PPRL solutions that can link data from multiple parties is that either (1) they only support exact matching (which classifies sets of records as matches if their masked QIDs are exactly the same and as non-matches otherwise) \cite{Kara15,Lai06} or (2) they are applicable to QIDs of categorical data only~\cite{Kan08,Moh11}. However, in many PPRL applications QIDs of string data, such as names and addresses, are required. These QIDs often contain errors and variations which necessitates the use of approximate comparison functions (that are computationally expensive in terms of the number of comparisons) for comparing QIDs.
In this paper, we tackle the multi-party PPRL problem by developing an efficient privacy-preserving approach for approximate matching of (masked) QIDs of string data from multiple records. 

We propose the use of Counting Bloom Filter (CBF) encoding, which is a variation of Bloom filter (BF) encoding~\cite{Vat13}, to enable efficient and approximate privacy-preserving linkage of multiple databases. 
BFs are bit vectors into which values are hash-mapped using hash functions (as we describe in Section~\ref{subsec-bf}). CBFs, on the other hand, are integer vectors that contain count values in each position. Multiple BFs can be summarized as a single CBF using the vector addition operation between BFs.
Previous BF encoding-based PPRL approaches~\cite{Dur12,Sch09} suggest using a linkage unit ($LU$), which is a dedicated external party that can perform linkage by comparing candidate record sets (masked into BFs) from all database owners and calculating their similarities to classify them as matches or non-matches.
Our hypothesis is that, rather than sending BFs from all parties to a linkage unit ($LU$) to calculate their similarity, a single CBF for each candidate record set generated over $p$ parties can be used to calculate their similarity, as illustrated in Figure~\ref{fig:PPRL_naive} (left for BFs and right for CBFs).
Since CBFs contain only the summary information (count values) of multiple records in their positions rather than the actual individual bit values of a single record as in BFs, they 
provide increased privacy compared to BFs, as we discuss in Section~\ref{sec-analysis}. To the best of our knowledge, this privacy aspect of CBFs has so far not been utilized in PPRL.

An additional challenge with multi-party PPRL is that complexity 
increases significantly with multiple parties in terms of both computational efforts and communication volume.
A basic approach would be to send all masked records from all $p$ parties to a $LU$ that can calculate pair-wise similarities between masked records, which is of $O(p^2 \cdot n^2)$ complexity, where $n$ is the size of databases assuming all databases are of equal size. However, identifying a matching set of records from all $p$ parties is not possible with such a basic pair-wise comparison approach.
On the other hand, the number of na\"{i}ve (all-to-all) comparisons between records required across all $p$ databases ($\mathbf{D}_1$, $\mathbf{D}_2$, $\cdots$, $\mathbf{D}_p$) is equal to the product of the size of the databases (i.e.\ $n \times n \times \cdots n = n^p$). 
Addressing this complexity challenge, two-step algorithms have been developed where in the first step a private indexing/blocking technique is used to reduce the number of candidate record sets from $n^p$ to $b \cdot n^p/b^p$, assuming $b$ blocks of equal size. In the second step only these candidate record pairs have to be compared and classified~\cite{Vat13}. 
Compared to the quadratic number of record pairs when linking only two databases ($O(n^2)$), in multi-party PPRL the number of candidate record sets increases exponentially with the number of parties ($O(n^p)$), and thus using existing private blocking techniques would not sufficiently reduce the number of comparisons, as has been empirically studied in several recent approaches~\cite{Ran14,Ran15,Vat14c}. 

Figure~\ref{fig:PPRL_naive} overviews the na\"{i}ve computation and communication of (masked) candidate record sets from multiple parties ($p=4$) using a $LU$. Independent of the used masking function (BFs in the left figure and CBFs in the right figure) and the communication pattern (direct one-to-one communication between each party and the $LU$ in the left figure and ring-based communication among the parties in the right figure, as will be described in Section~\ref{sec:precon}), the na\"{i}ve approach results in exponential complexity.
Efficient communication patterns and advanced filtering approaches for multi-party PPRL therefore need to be developed in order to reduce the potentially huge number of comparisons. 
Moreover, with multiple parties the privacy risk of collusion increases, where a sub-set of parties collude among them in order to learn about another party's (or sub-set of parties') private data. Both examples for the na\"{i}ve method described in Figure~\ref{fig:PPRL_naive} are highly susceptible to collusion. 

In order to overcome these scalability and privacy challenges of multi-party PPRL, we introduce two efficient CBF-based communication patterns that either use a $LU$ or operate symmetrically without a $LU$ (where a trusted external party is not available to act as a $LU$). The proposed approaches can significantly reduce the number of comparisons required between records in contrast to the na\"{i}ve all-to-all comparisons, and thereby improve the scalability while also improving the privacy (reducing the likelihood of collusions) by arranging parties into several groups and by distributing computations among parties.

\smallskip
\textbf{Contributions:}
Our contributions in this paper are: (1) a novel multi-party PPRL protocol based on CBFs and secure summation for efficient, approximate, and private linkage; (2) two variations of extended secure summation protocols for improved privacy against collusion among the data base owners: (a) homomorphic encryption-based and (b) salting-based (using random seed integers); (3) two efficient communication patterns (with and without a $LU$) for reducing the comparison space and risk of collusion between parties and thereby improving scalability and privacy, respectively, in multi-party PPRL; (4) an analysis of the protocol in terms of the three properties of PPRL: scalability (complexity), linkage quality, and privacy; and (5) an empirical evaluation and comparison of our protocol with two baseline approaches using large North Carolina Voter Registration (NCVR)~\cite{Chr13NC} datasets.

\smallskip
%The remainder of the paper is structured as follows:
\textbf{Outline:} 
In the following section 
we describe the preliminaries. 
In Section~\ref{sec:algorithm} we 
propose our protocol
for multi-party PPRL 
based on CBFs and secure summation,  
where in Section~\ref{sec-ext-secsum} we we propose two extended
secure summation protocols to improve privacy of our approach, 
and in 
Section~\ref{sec-comm-patterns} we introduce two efficient
communication patterns to improve scalability and privacy. 
We analyze our protocol 
in terms of complexity, linkage quality, and privacy
in Section~\ref{sec-analysis}, and  
in Section~\ref{sec-experiment} we
conduct an empirical study on the NCVR datasets
to validate these analyses. We provide a review
of related work in Section~\ref{sec-related}.
Finally, we summarize and
discuss future research directions in
Section~\ref{sec:con}.

% --------------------------------------------------------------------

\section{Preliminary concepts and building blocks}
\label{sec:precon}

In this section, we define 
the problem of multi-party PPRL and explain
how CBFs can be used
for efficiently calculating similarities (approximate matching) of 
QID values 
between a set of multiple (two or more) 
records (held by different parties) in PPRL. 

We assume $p$ database owners $P_1$, $P_2$, $\cdots$, $P_p$
with their respective databases $\mathbf{D}_1$, $\mathbf{D}_2$, $\cdots$, $\mathbf{D}_p$
(containing sensitive or
confidential identifying information)
participate in the process under the honest-but-curious (HBC)~\cite{Vat13}.
In the HBC model, parties are assumed to follow the protocol without
deviating or sending false information while being curious to learn about other
parties' data.
However, the HBC model does not assume that the parties do not collude
among them to learn about other parties' data~\cite{Lin09}.
We quantify the risk of collusion in multi-party PPRL and the reduction
of risk by our communication patterns in Section~\ref{subsec_privacy_analsis}.
We also assume a set of QID attributes $A$, which will
be used for the linkage, is common to all these databases.
We formally define the problem of PPRL of multiple databases as follows.

\begin{defi}
\textbf{Multi-party PPRL}: 
Assume $P_{1}, \ldots, P_{p}$ are the $p$ owners (parties) of the databases $\mathbf{D}_{1}, \ldots, \mathbf{D}_{p}$, respectively. They wish to determine which of their records $R_{1,i} \in \mathbf{D}_{1}$, $R_{2,j} \in \mathbf{D}_{2}$, $ \ldots$,  $R_{p,k} \in \mathbf{D}_{p}$ match based on the (masked) QIDs of these records according to a decision model $C(R_{1,i}$,$R_{2,j}$, $\ldots$, $R_{p,k})$ that classifies record sets $(R_{1,i}$,$R_{2,j}$, $\ldots$, $R_{p,k})$ into one of the two classes $\mathbf{M}$ of matches and $\mathbf{U}$ of non-matches. Assuming the HBC adversary model, parties $P_{1}, \ldots, P_{p}$ are honest, in that they follow the linkage protocol steps, while they do not wish to reveal their actual records ${R}_{1,i}, \ldots, {R}_{p,k}$ with any other party. They however are prepared to disclose to each other, or to an external party, the actual values of some selected attributes of the record sets that are in class $\mathbf{M}$ to allow analysis. 
\end{defi}

Masking functions used for privacy-preserving algorithms can be 
categorized into two: cryptographic-based secure multi-party computation (SMC)
techniques and perturbation-based techniques~\cite{Vat14}.
The former approach is generally more expensive
with regard to the computation and communication complexities
though it provides strong privacy guarantees and high accuracy~\cite{Lin00}. 
The latter uses efficient techniques and, as opposed to SMC techniques, 
these techniques aim to hide (mask) information about the original
values (to preserve privacy) 
while still allowing to perform approximate matching between the
masked values using the functional relationship between original
and masked data.

We propose an efficient protocol for multi-party PPRL 
using perturbation-based masking.
In this section, we describe the four building blocks of our protocol, 
and in Section~\ref{sec:algorithm} we present our algorithm in detail.
In the following two sections we assume a linkage unit
($LU$) is available to conduct the linkage, 
and in Section~\ref{sec-comm-patterns} we propose a variation where
a $LU$ is not required to conduct the linkage using our protocol.

% - - - - - - - - - - - - - - - - - - - - - - - - - - - - - - - - - -

\subsection{Bloom filter encoding}
\label{subsec-bf}
Bloom filter (BF) encoding has been used as an efficient masking (encoding) technique
in several PPRL solutions~\cite{Dur13,Lai06,Ran14,Sch11,Seh15,Vat12,Vat14c,Vat16c}.
A BF $b_i$ is a bit array data structure of length $l$ bits
where all bits are initially set to $0$. $k$ independent hash
functions, $h_1,h_2, \ldots, h_k$, each with range $1, \ldots l$, are
used to map each of the elements $s$ in a set $S$ into the BF by
setting the bit positions $h_j(s)$ with $1 \le j \le k$ to $1$. 

Schnell et al.~\cite{Sch11} were the first to propose a method for
approximate matching in PPRL of two databases using BFs. In
their work, as in our protocol, the character $q$-grams (sub-strings of length
$q$) of QID values in $A$ of each record in the databases to be
linked are hash-mapped into a BF using $k$ independent hash
functions. 
This method of BF encoding is 
known as Cryptographic Long term Key (CLK) encoding~\cite{Sch11}.

These BFs are then either sent to a $LU$ that
calculates the similarity of pairs of
BFs, as suggested by Schnell et al.~\cite{Sch11} and Durham et al.~\cite{Dur13},
or they are partially exchanged among the database owners to distributively
calculate the similarities of BF pairs/sets, as proposed by
Lai et al.~\cite{Lai06} and Vatsalan and Christen
for two-party~\cite{Vat12} and multi-party~\cite{Vat14c,Vat16c} approaches.
Figure~\ref{fig:bloomfilter}(a) illustrates the encoding of
bigrams ($q=2$) of two QID values `peter' and `pete' into
$l=14$ bits long BFs using $k=2$ hash functions.

\begin{figure}[!t]
  \centering
  \scalebox{1.0}[1.0]{\includegraphics[width=0.47\textwidth]
                      {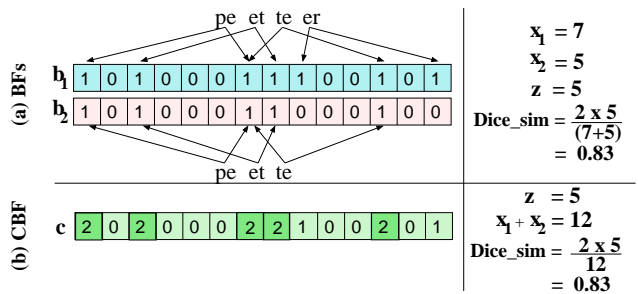}}
  \caption{\small{Similarity calculation of QIDs of two records masked using (a) BFs and (b) CBFs,
                  where $l=14$, $k=2$, and $q=2$.}}
           \label{fig:bloomfilter}
\end{figure}

\subsection{Dice coefficient}
\label{subsec-dc}
Any set-based similarity function
(such as overlap, Jaccard, and Dice coefficient) can be used to calculate the
similarity of pairs or sets of (multiple) BFs. The Dice coefficient has
been used for matching of BFs since it is insensitive to
many matching zeros (bit positions to which no $q$-grams are hash-mapped) in long BFs~\cite{Sch11}.
In future, we aim to investigate how other approximate string comparison functions, 
such as edit distance~\cite{Ris98} and Jaro and Winkler~\cite{Jaro89,Win90}, 
can be extended to calculate
the similarity of more than two values.

\begin{defi}
We define the Dice coefficient similarity of $p$ BFs ($b_1, \cdots,
b_p$) as:
\begin{eqnarray}
\label{eq:Dice_coefficient}
Dice\_sim (b_1, \cdots, b_p) &=& \frac{p \times z}{\sum_{i=1}^{p} x_i} 
\end{eqnarray}
where $z$ is the number of common bit positions that are set to $1$ in
all $p$ BFs (common $1$-bits), and $x_i$ is the number
of bit positions set to $1$ in $b_i$ ($1$-bits), $1 \le i \le p$.
\end{defi}

Figure~\ref{fig:bloomfilter}(a) illustrates the Dice coefficient 
similarity calculation of the two QID values `peter' and `pete'
masked into BFs.

\subsection{Secure summation}
\label{subsec-secsum}
A secure summation protocol~\cite{Clif02} can be used
to securely sum the input values of multiple 
parties ($p \ge 3$), such that no party learns the individual values
of other parties, but only the summed value.
The input can either be a single numeric value or
a vector of numeric values. 
For example, $p$ numeric values ($v_1,v_2,\cdots,v_p$) can be securely summed
by using a random numeric value $r'$ which is sent by a $LU$.
The first party $P_1$ that receives $r'$ calculates the summation
of $r'+ v_1$ and sends to $P_2$. This process is repeated until
the last party sends $r'+v+1+v_2+\cdots+v_p$ to the $LU$ which then
subtracts $r'$ from the summed value to calculate the summation of
$p$ values.
The protocol employs a ring-based communication pattern over all parties 
which allows the $LU$ to learn the final values ($\sum_{i=1}^p v_i$) 
but no party will learn the
values $v_i$ of the other parties.
This basic secure summation (\emph{BSS}) protocol is susceptible to
collusion risk among the parties. In Section~\ref{sec-ext-secsum}
we propose extended secure summation protocols for improved
privacy against collusion.

\subsection{Counting Bloom filters (CBFs)}
\label{subsec-cbf}

In this section we propose a novel method of calculating the similarity
of $p$ BFs using a CBF, which will provide increased privacy
compared to using BFs for similarity calculation, as we will
discuss in Section~\ref{subsec_privacy_analsis}.
A CBF $c$ of $p$ ($p > 1$) BFs is an integer array
data structure of the same length as the BFs ($l$).
It contains the counts
of values in each bit position $\beta, 1 \le \beta \le l$ over a set of $p$ BFs
in its corresponding position, such that $c[\beta] = \sum_{i=1}^p b_i[\beta]$, where
$c[\beta]$ is the count value in the $\beta$ bit position of the CBF $c$ and
$b_i[\beta] \in[0,1]$ provides the value in the bit position $\beta$ of BF $b_i$.
Given $p$ BFs (bit vectors) $b_i$ with $1 \le i \le p$, the CBF
$c$ can be generated by applying a vector addition operation
between the bit vectors such that $c = \sum_i b_i$.

\begin{theorem}
The Dice coefficient similarity of $p$ BFs can be calculated given
only a CBF as:
\begin{eqnarray}
\label{eq:Dice_coefficient_cbf} 
Dice\_sim (b_1, \cdots, b_p) = \frac{p \times |\{\beta: c[\beta] = p, 1 \le \beta \le l\}|}{\sum_{\beta=1}^{l} c[\beta]} \nonumber
\\
\end{eqnarray}
\end{theorem}

\begin{proof}
The Dice coefficient similarity of $p$ BFs ($b_1$,$b_2$, $\cdots$, $b_p$) is determined by the sum of $1$-bits ($\sum_{i=1}^p x_i$) in the denominator of Eq.~(\ref{eq:Dice_coefficient}) and the number of common $1$-bits ($z$) in all $p$ BFs in the nominator of Eq.~(\ref{eq:Dice_coefficient}).
The number of $1$-bits in a BF $b_i$ is $x_i = b_i[1] + b_i[2] + \cdots + b_i[l]$, with $1 \le i \le p$. The sum of $1$-bits in all $p$ BFs is therefore $\sum_{i=1}^p x_i = \sum_{i=1}^p b_i[1] + b_i[2] + \cdots + b_i[l]$. The value in a bit position $\beta$ ($1 \le \beta \le l$) of the CBF of these $p$ BFs is $c[\beta] = b_1[\beta] + b_2[\beta] + \cdots + b_p[\beta]$. The sum of values in all bit positions of the CBF is $\sum_{\beta=1}^{l} c[\beta] = \sum_{\beta=1}^{l} b_1[\beta] + b_2[\beta] + \cdots + b_p[\beta]$ which is equal to $\sum_{i=1}^p x_i = \sum_{i=1}^p b_i[1] + b_i[2] + \cdots + b_i[l]$. 
Further, if a bit position $\beta$ ($1 \le \beta \le l$) contains $1$ in all $p$ BFs, i.e.\ $\forall_{i=1}^p b_i[\beta] = 1$, then $c[\beta] = \sum_{i=1}^p b_i[\beta] = p$. Therefore, the common $1$-bits ($z$) that occur in all $p$ BFs can be calculated by counting the number of positions $\beta \in c$ where $c[\beta] = p$, while the sum of number of $1$-bits ($\sum_{i=1}^p x_i$) is calculated by summing the values in bit positions $\beta \in c$, $\sum_{\beta=1}^l c[\beta]$.
\end{proof}

If the $LU$ gets only the CBF $c$
that contains the summed values in the
bit positions of $p$ BFs instead of the actual $p$ BFs, then
the $LU$ can calculate the Dice coefficient of $p$ BFs
using Eq.~(\ref{eq:Dice_coefficient_cbf}) without
learning any information about the individual bit positions
of each party. 
As an example, the Dice coefficient
calculation of the two BFs from the 
two parties $P_1$ and $P_2$ in Figure~\ref{fig:bloomfilter}(a)
is extended by
using a CBF and secure summation 
to calculate the similarity 
by the $LU$, as shown
in Figure~\ref{fig:bloomfilter}(b). 
The information gained by the $LU$ (and/or database owners) 
from a single CBF is less than $p$ BFs
(i.e.\ CBFs provide increased privacy compared to BFs), as
theoretically proven in Section~\ref{subsec_privacy_analsis} and empirically
validated in Section~\ref{sec-experiment}.

\section{Multi-party PPRL algorithm}
\label{sec:algorithm}

Our protocol allows efficient, approximate, and private linking of
records based on their masked QID values in multiple databases
from $p$ $(\ge 2)$ sources/parties.
We first describe a basic na\"{i}ve protocol
based on CBFs,
which we name as `NAI', and in Section~\ref{sec-comm-patterns} we then
propose improved communication patterns for this protocol
to make it more scalable.
The steps of our protocol are listed below. 

\begin{itemize}

\item \textbf{Step 1:}~
The parties agree upon the following parameter values:
      the BF length $l$;
      the $k$ hashing 
      functions $h_1, \ldots, h_k$ to be used; the length (in characters) of 
      grams $q$; a minimum similarity threshold value $s_t$ ($0 \le s_t \le 1$),
      above which a set of records is classified as a match;
      a private blocking function $block(\cdot)$; 
      the blocking keys~\cite{Chr12} $B$ used for blocking; 
      and a set of QID attributes $A$ used for the linkage.

\item \textbf{Step 2:}~
Each party $P_i$ ($1 \le i \le p$) individually 
applies a private blocking
function~\cite{Vat13} $block(\cdot)$ 
to reduce the number of candidate sets of
records $\mathbf{C}$ (from $\prod_i^p n_i$, where $n_i = |\mathbf{D}_i|$ 
is the number of records in $\mathbf{D}_i$ held by party $P_i$),
which otherwise becomes prohibitive even
for moderate $n_i$ and $p$. 
The $block(\cdot)$ function groups records according to their
blocking key values (BKVs)~\cite{Chr11} and only 
records with the same BKV 
(i.e. records in the same block)
from different parties are then 
compared and classified 
using our protocol.
A phonetic-based blocking~\cite{Chr12}
or any of the 
existing multi-party blocking techniques for PPRL~\cite{Ran14,Ran15}
can be used as the $block(\cdot)$ function. 

\begin{table}[!t]
\scriptsize\addtolength{\tabcolsep}{-3pt}
\begin{tabular*}{0.46\textwidth}{c @{\extracolsep{\fill}} lll}
  \label{algo_sec_sim_calc}
    ~ \\[2mm] \hline 
    \\[-2mm]
    \multicolumn{3}{l}{\textbf{Algorithm~1:} Secure summation of BF vectors.}
      \\[0.5mm] \hline
    ~ \\[-3mm]
    \multicolumn{3}{l}{\textbf{Input:}} \\
    {- $\mathbf{D}_i^M$:} & 
      \multicolumn{2}{l}{Party $P_i$'s records masked into BFs $b_i$, $1 \le i \le p$} \\
      {- $\mathbf{R'}$:} & \multicolumn{2}{l}{List of random vectors used by party $P_1$ or the $LU$} \\
      ~ & \multicolumn{2}{l}{for secure summation of BFs} \\
    \multicolumn{3}{l}{\textbf{Output:}} \\
    {- $\mathbf{C}$:} & \multicolumn{2}{l}{Candidate record sets with CBFs} \\[1mm]
    1:&$\mathbf{C} = \{\,\}$ &//initialize $\mathbf{C}$ \\ 
    2:&\textbf{for} $1 \le i \le p$ \textbf{do}: & ~ \\
    3:&~~ \textbf{if} $i = 1$ \textbf{then}:&//first party\\
    4:&~~ ~~ \textbf{for} $rec \in \mathbf{D}_1^M$ \textbf{do}: & ~ \\
    5:&~~ ~~ ~~ $c = \mathbf{R'}[rec] + \mathbf{D}_1^M[rec]$ &//summation \\ 
    6:&~~ ~~ ~~ $\mathbf{C}[rec] = c$ & ~ \\ 
    7:&~~ ~~ $\mathbf{C}.send\_to(P_2)$ &//send $\mathbf{C}$ to $P_2$\\
    8:&~~ \textbf{else}: &//other parties \\
    9:&~~ ~~ $\mathbf{C}.receive\_from(P_{i-1})$ &//receive from $P_{i-1}$\\
    10:&~~ ~~ \textbf{for} $cs \in \mathbf{C}$ \textbf{do}: & ~ \\
    11:&~~ ~~ ~~ \textbf{for} $rec \in \mathbf{D}_i^M$ \textbf{do}: & ~ \\
    12:&~~ ~~ ~~ ~~ $c = \mathbf{C}[cs]$ +$ \mathbf{D}_i^M[rec]$ &//addition \\ 
    13:&~~ ~~ ~~ ~~ $cs = cs.append(rec)$ & \\
    13:&~~ ~~ ~~ ~~ $\mathbf{C}[cs]$ $= c$ & ~ \\ 
    14:&~~ ~~ \textbf{if} $i \neq p$ \textbf{then}: & ~ \\ 
    15:&~~ ~~ ~~ $\mathbf{C}.send\_to(P_{i+1})$ &//send $\mathbf{C}$ to $P_{i+1}$\\ 
    16:&~~ ~~ \textbf{else}: & ~ \\ 
    17:&~~ ~~ ~~ $\mathbf{C}.send\_to(P_{1}~or~LU)$ &//send $\mathbf{C}$ to $P_1$/$LU$\\ [1mm]
      \hline \\ [2mm]
  \end{tabular*}
\end{table}

\item \textbf{Step 3:}~
Each party $P_i$
hash-maps the $q$-gram values of QIDs $A$ of each of
its $n_i$ records in their respective databases $\mathbf{D}_i$
into $n_i$ BFs (one BF per record in $\mathbf{D}_i$) 
of length $l$ using
the hash functions $h_1, \ldots, h_k$. It is crucial
to set the BF related parameters in an optimal
way that balances all three properties of PPRL 
(complexity, quality, and privacy). We further discuss parameter
setting for BFs used in our protocol
in Section~\ref{sec-analysis}.

\item \textbf{Step 4:}~
In the next step, the parties initiate a secure summation protocol
to securely perform vector addition between their BFs in order
to generate a CBF for each set of candidate records $\mathbf{C}$.
This secure summation can be initiated 
by an external linkage unit $LU$ that provides a vector (of length $l$)
of random values 
or by one of the parties (as will be discussed 
further in Section~\ref{sec-comm-patterns}). 
This process 
is outlined in Algorithm~1.

In lines~3-6, the $LU$ or the party that initiated the communication (we assume $P_1$)
sends or uses, respectively, a random vector $\mathbf{R'}[rec]$
for each record $rec \in \mathbf{D}_1^M$
to sum
with the party's BF vector $b_i \in \mathbf{D}_i^M$ ($i=1$) 
by applying a vector addition operation.
The summed vectors $\mathbf{R'}[rec]+b_i \in \mathbf{C}$  
are then sent
to party $P_{i+1}$ 
in line~7.
Party $P_{i}$, $2 \le i \le p$ receives the summed vectors from $P_{i-1}$ (line~9)
and adds its BF vector $b_i \in \mathbf{D}_i^M$ to each candidate set
$cs \in \mathbf{C}$ 
and sends the summed vectors to the next party $P_{i+1}$.
This process is repeated until the last party (i.e.\ $P_p$)
adds its $b_p$ vector to each received summed vector
$\mathbf{R'}[rec] + \sum_{i=1}^{p-1} b_i$ from
party $P_{p-1}$ and sends the final summed vector back to the $LU$ or $P_1$
for each candidate set $cs$, 
as shown in lines~8-17 in Algorithm~1. 

\begin{table}[!t]
\scriptsize\addtolength{\tabcolsep}{-3pt}
\begin{tabular*}{0.45\textwidth}{c @{\extracolsep{\fill}} lll}
  \label{algo_classify}
    ~ \\[2mm] \hline 
    \\[-2mm]
    \multicolumn{3}{l}{\textbf{Algorithm~2:} Similarity calculation of record sets.} %(by $P_1$ or $LU$)
      \\[0.5mm] \hline
    ~ \\[-3mm]
    \multicolumn{3}{l}{\textbf{Input:}} \\
      {- $\mathbf{R'}$:} & \multicolumn{2}{l}{ List of random vectors used by party $P_1$ or the $LU$ for} \\
      ~ & \multicolumn{2}{l}{ secure summation of BFs} \\
    {- $\mathbf{C}$:} & 
      \multicolumn{2}{l}{ Candidate record sets with CBFs} \\
     ~ & \multicolumn{2}{l}{ (including random vectors) from party $P_{p}$} \\
    {- $s_t$:} & \multicolumn{2}{l}{ Minimum similarity threshold to classify record sets} \\
    \multicolumn{3}{l}{\textbf{Output:}} \\
    {- $\mathbf{M}$:} & \multicolumn{2}{l}{ Matching record sets} \\[1mm]
    1:&$\mathbf{M} = \{\,\}$&//initialize $\mathbf{M}$\\ 
    2:&$\mathbf{C}.receive\_from(P_p)$&//receive CBFs\\
    3:&\textbf{for} $cs \in \mathbf{C}$ \textbf{do}: & ~ \\
    4:&~~ $c = \mathbf{C}[cs] - \mathbf{R'}[cs]$&//subtract\\ %vector subtraction
    5:&~~ $Dice\_sim(cs) = \frac{p \times |\{\beta: c[\beta] = p, 1 \le \beta \le l\}|}{\sum_{\beta=1}^{l} c[\beta]}$&//Eq.~(\ref{eq:Dice_coefficient_cbf})\\ 
    6:&~~ \textbf{if} $Dice\_sim(cs) \ge s_t$ \textbf{then}:&//a match \\  %//classify as a match
    7:&~~ ~~ $\mathbf{M}.append([cs,Dice\_sim(cs)])$ & ~ \\
    8.&\textbf{for} $1 \le i \le p$ \textbf{do}: & ~ \\
    9:&~~ $\mathbf{M}.send\_to(P_i)$&~\\ [1mm] %//send $\mathbf{M}$ to all parties
      \hline \\ [2mm]
  \end{tabular*}
\end{table}

\item \textbf{Step 5:}~
Finally, either the $LU$ or the first party, $P_1$, 
that has provided the random vectors $\mathbf{R'}$,  
subtracts $\mathbf{R'}[rec]$ from the final summed vector
$\mathbf{R'}[rec] + \sum_{i=1}^{p} b_i$ as received from
the last party $P_p$
for each candidate set $cs \in \mathbf{C}$. 
This is achieved by subtraction of vectors
(i.e. $c =\mathbf{R'}[rec] + \sum_{i=1}^{p} b_i - \mathbf{R'}[rec]$), which is a special case of
vector addition, as outlined in lines~2-4 in Algorithm~2. 
As shown in lines~5-7, the $LU$ or $P_1$ then 
calculates the Dice coefficient similarity of each resulting CBF
      $c$ following
      Eq.~(\ref{eq:Dice_coefficient_cbf}) to classify the compared sets
      of records (BFs) within a block 
      into matches and non-matches based on
      the similarity threshold $s_t$.
      The final similarities of matching sets of records 
      will be sent to all parties $P_i$, with $1 \le i \le p$, in lines~8-9
      in Algorithm~2.

\end{itemize}

The basic secure summation protocol (\emph{BSS}) used in this 'NAI' protocol
is vulnerable to collusion risk among the parties.
Further, the number of candidate sets to be compared for multi-party linkage
in this na\"{i}ve method (\emph{NAI}) is 
exponential in the number of parties and their dataset sizes,
which is prohibitively large to be practically feasible
(even with the existing private blocking and filtering approaches
employed)~\cite{Vat14c}.
Efficient communication patterns among the parties
therefore need to be employed in order to make multi-party PPRL
scalable and viable in real applications that require data of very large
sizes from many parties to be integrated. 
In the following sections we propose extended secure summation protocols
and two efficient communication patterns 
that not only improve the scalability of our multi-party PPRL protocol 
but also make it more secure (with less possible collusion among the parties).

\begin{figure*}[!t]
  \centering
  \scalebox{0.9}[0.8]{\includegraphics[width=1.0\textwidth]
                      {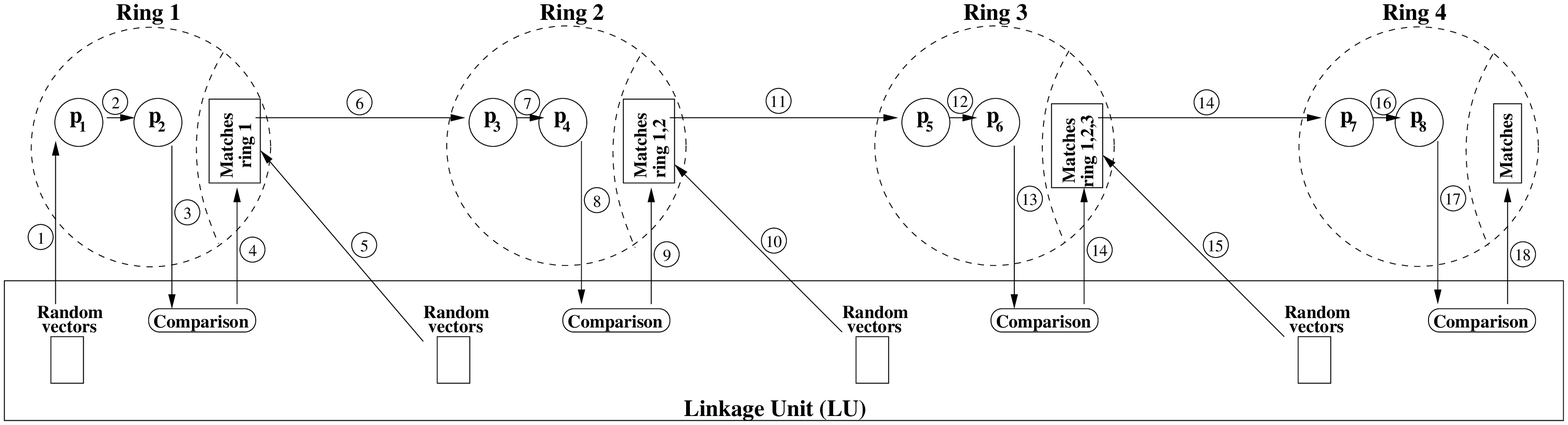}}
  \caption{\small{Sequential (\emph{SEQ}) communication pattern for comparison of CBFs from $p=8$ parties 
using a $LU$, as described in Section~\ref{subsec-seq_comm}.
In this example, the minimum number of parties per ring is set to $r_m=2$. Communication steps are shown as circled numbers.}
}
           \label{fig:cbf_seq}
\end{figure*}

% - - - - - - - - - - - - - - - - - - - - - - - - - - - - - - - - - -
\section{Extended secure summation}
\label{sec-ext-secsum}

The basic secure summation protocol (\emph{BSS}) described
in Section~\ref{subsec-secsum} is susceptible
to collusion risk by the parties involved, where if two or more parties
collude they are able to infer the input of another
party. In order to overcome this risk (to improve privacy), we propose
two extended secure summation protocols.

\begin{itemize}

\item Homomorphic-based secure summation (\emph{HSS}):
The partially homomorphic Paillier cryptosystem~\cite{Pai99}
is a reliable secure multi-party computation (SMC) technique
for performing secure joint computation among several parties.
In the \emph{HSS} approach a pair of private and public keys is used
for encrypting and decrypting the individual BF values.
Successive encryption of the same value using the same public key
generates different encrypted values with high probability, 
and decrypting the encrypted values using a private key returns
the correct original value.
The public key is known to all parties while the private key is
known only to the $LU$. Each party $P_i$ receives
summed vectors containing encrypted values to which $P_i$ adds 
its encrypted $\mathbf{b}_i$ vector (using the public key) and sends the
encrypted summed vectors to the next party. Without knowing the private key
a party $P_i$ cannot decrypt the received vectors and therefore 
colluding with a party to learn another party's $\mathbf{b}_j$ (with $i \neq j$)
would be impossible.

\item Salting-based (using random seed integers) secure summation (\emph{SSS}):
The \emph{HSS} approach provides a secure solution 
compared to the \emph{BSS} approach
against collusion attacks at the cost of 
an excessive computation overhead, 
making it not scalable to linking multiple large databases.
Therefore, we propose the \emph{SSS} approach 
to provide security against collusion attacks in an efficient way.
Salting has been used to defend against dictionary attacks
on one-way hash functions where an additional string 
is concatenated with
a value to be encrypted~\cite{Sch15}.
We adopt a similar concept in the \emph{SSS} approach where the salting key is
an additional random integer used by each party $P_i$ individually 
when performing the secure summation. The salting key generated and 
used by each party is sent only
to the $LU$ and therefore a party $P_i$'s BF values cannot be identified
by means of collusion among the parties, as without knowing the salting
key of $P_i$ its BF values cannot be learned.
Since the salting keys are random integer values, performing secure summation
of BFs with the salting keys does not add any additional 
computation and communication overhead, except the communication of
salting keys from parties to the $LU$.

\end{itemize}

% - - - - - - - - - - - - - - - - - - - - - - - - - - - - - - - - - -

\section{Communication patterns}
\label{sec-comm-patterns}

In this section, we propose two
variations of improved communication patterns for 
our protocol based on CBFs
for multi-party scenarios 
with and without a linkage unit ($LU$).
The main idea of these improved communication patterns is to
exploit the facts that most candidate sets are true
non-matches (due to the class imbalance problem of record 
linkage~\cite{Chr11}), and that a true matching set must have a
high similarity between any sub-set of records in that set.
Hence for multi-party PPRL applications with
many parties it is promising to determine 
partial matches for a sub-set of parties and consider 
additional parties only for these partial matches.

The parties are first arranged into rings of size $r$ (with $r \le p$)
based on the value for the parameter $r_m$, 
the minimum number of parties per ring ($r \ge r_m$). The value of $r_m$ needs
to be carefully chosen, as it has a trade-off between scalability (complexity) and
privacy. The higher the value for $r_m$ is, the better the privacy of the
protocol because the resulting CBFs are more difficult
to exploit with an inference attack (by mapping the CBFs to known
values in a global database to infer their underlying unencoded values),
as will be discussed in Section~\ref{subsec_privacy_analsis}.
On the other hand, higher values of $r_m$
results in lower scalability to large datasets across many parties
because the number of comparisons required per ring exponentially increases
with the ring size $r$.

\begin{figure*}[!t]
  \centering
  \scalebox{0.8}[0.6]{\includegraphics[width=1.0\textwidth]
                      {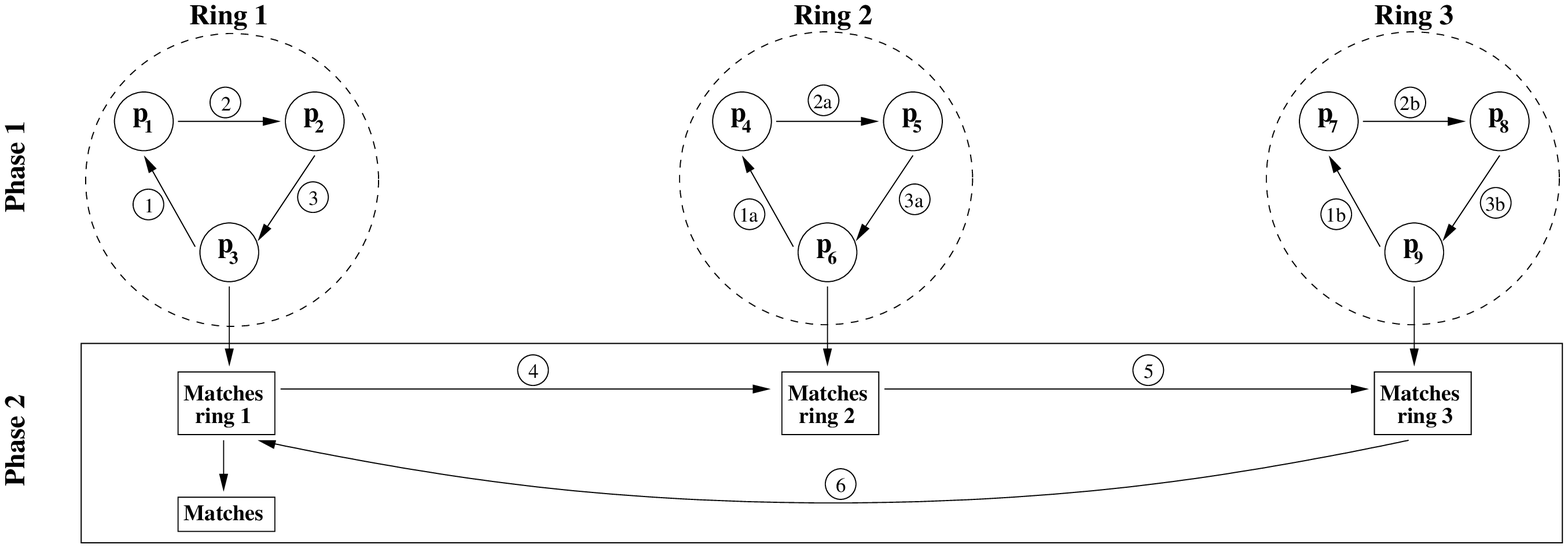}}
  \caption{\small{Ring by ring (\emph{RBR}) communication pattern for comparison of CBFs from $p=9$ parties 
without a $LU$, as described in Section~\ref{subsec-rbr_comm}.
In this example, the minimum number of parties per ring is $r_m=3$. Communication steps are shown as circled numbers.}
}
           \label{fig:cbf_rbr}
\end{figure*}

\subsection{Sequential communication}
\label{subsec-seq_comm}

In this first proposed communication pattern (\emph{SEQ}), which requires a $LU$, 
the matches found in
one ring are compared with the candidate sets in the next ring
resulting in a set of matches from both rings which will then be
compared with the candidate sets in the following ring, and so on. This communication
is carried out sequentially until the matches 
from all rings are found by the $LU$.
Figure~\ref{fig:cbf_seq} illustrates the \emph{SEQ}
approach for four rings with $r=2$ parties in each ring.

\begin{table}[!t]
\scriptsize\addtolength{\tabcolsep}{-3pt}
\begin{tabular*}{0.47\textwidth}{c @{\extracolsep{\fill}} lll} 
  \label{algo_classify}
    ~ \\[2mm] \hline 
    \\[-2mm]
    \multicolumn{3}{l}{\textbf{Algorithm~3:} Comparison of CBFs using \emph{SEQ} (by $LU$).}
      \\[0.5mm] \hline
    ~ \\[-3mm]
    \multicolumn{3}{l}{\textbf{Input:}} \\
    {- $\mathbf{D}^M_i$:} & \multicolumn{2}{l}{ Party $P_i$'s records with RIDs and BFs, $1 \le i \le p$} \\
    {- $r_m$:} & \multicolumn{2}{l}{ Minimum ring size, with $r_m \ge 2$} \\
    {- $s_t$:} & \multicolumn{2}{l}{ Minimum similarity threshold to classify record sets} \\
    \multicolumn{3}{l}{\textbf{Output:}} \\
    {- $\mathbf{M}$:} & \multicolumn{2}{l}{ Matching record sets} \\[1mm]
    1:&$\mathbf{C} = \{[\,]\}$; $\mathbf{M} = \{[\,]\}$&//initialize $\mathbf{C}$, $\mathbf{M}$\\
    2:&$rings = group([P_1,P_2,\cdots, P_p],r_m)$&//group parties \\
    3:&\textbf{for} $ring \in rings$ \textbf{do}:&//iterate rings\\
    4:&~~ \textbf{for} $i \in ring$ \textbf{do}:&//iterate parties \\    
    5:&~~ ~~ $records = \mathbf{D}^M_i.values()$&//party $P_i$'s records \\
    6:&~~ ~~ \textbf{for} $rec\_set \in \mathbf{C}$ \textbf{do}: & ~ \\
    7:&~~ ~~ ~~ \textbf{for} $rec \in records$ \textbf{do}: & ~ \\
    8:&~~ ~~ ~~ ~~ $rec\_set.append(rec) $&//$\mathbf{C}$ of the ring\\
    9:&~~ $\mathbf{C'} = sec\_sum(\mathbf{C},ring)$&//summation\\
    10:&~~ \textbf{for} $cs \in \mathbf{C'}$ \textbf{do}: & ~ \\
    11:&~~ ~~ $c = \mathbf{C'}[cs]$ & ~ \\
    12:&~~ ~~ $Dice\_sim(cs) = \frac{p \times |\{\beta: c[\beta] = p, 1 \le \beta \le l\}|}{\sum_{\beta=1}^{l} c[\beta]}$ &//Eq.~(\ref{eq:Dice_coefficient_cbf})\\ 
    13:&~~ ~~ \textbf{if} $Dice\_sim(cs) \ge s_t$ \textbf{then}:&//a match\\
    14:&~~ ~~ ~~ $\mathbf{M}.append(cs)$ & ~ \\ 
    15:&~~ $\mathbf{C} = \mathbf{M}$&~\\ [1mm] %//for the next ring
      \hline \\ [2mm]
  \end{tabular*}
\end{table}

Algorithm~3 details the steps of the \emph{SEQ} communication pattern.
First the parties are grouped in rings using the function 
$group(\cdot)$ (line~2 in Algorithm~3).
Different number of parties ($\ge r_m$) can be grouped into different rings.
The value for $r_m$ can be agreed upon by the parties to any value $r_m \ge 2$
at the trade-off between privacy and scalability, as will be discussed
in Section~\ref{sec-analysis}.

To minimize the number of comparisons that are required, 
the grouping of parties into rings is ideally done in an ascending order
of the size of their datasets.
In this way, the first ring will generate a smaller number of matches,
which then have to be compared with the candidate sets
in the following rings.
This reduces the computational complexity compared to 
initially larger number of matches being generated by the 
first ring if the parties in the first ring are the ones with 
the largest databases.

A loop is iterated over rings in line~3 of Algorithm~3 and then the
parties in rings are iterated in line~4.
Each party in $ring$ retrieves its records in line~5 and a loop is
iterated over these records in line~7 to append each of them to
every candidate record set (from previous party, except for the first party
that appends to empty sets) stored in $\mathbf{C}$ (line~8).
Next, a secure summation protocol is applied in line~9 using 
$sec\_sum(\cdot)$ function
on the candidate sets of BFs $\mathbf{C}$ identified in $ring$
in order to generate CBFs for each candidate set. In lines~10-12, each CBF $c$ generated
is then used to calculate the Dice similarity ($Dice\_sim$) of the candidate 
set ($cs$) and if the $Dice\_sim(cs) \ge s_t$ (line~13) then $cs$ is added into
the list of matches $\mathbf{M}$
identified in the ring (line~14), which will then be used as an input (line~15)
in the next ring.

The risk of collusion between parties 
in order to identify data about another party can be reduced in
this approach by using different BF encodings
in different iterations. For example, if the encoding
used in ring 1 in Figure~\ref{fig:cbf_seq} is unknown to
parties in ring 2, then the collusion between the $LU$ and parties in ring
2 would not reveal sufficient information to infer the
actual values masked in the BFs in ring 1. 
Hence, parties might wish to be grouped with certain other parties in the
same ring for better privacy protection. 

\begin{table}[!t]
\scriptsize\addtolength{\tabcolsep}{-6pt}
\begin{tabular*}{0.47\textwidth}{c @{\extracolsep{\fill}} lll}
  \label{algo_classify}
    ~ \\[2mm] \hline 
    \\[-2mm]
    \multicolumn{3}{l}{\textbf{Algorithm~4:} Comparison of CBFs using \emph{RBR} (without $LU$).}
      \\[0.5mm] \hline
    ~ \\[-3mm]
    \multicolumn{3}{l}{\textbf{Input:}} \\
    {- $\mathbf{D}^M_i$:} & \multicolumn{2}{l}{ Party $P_i$'s records with RIDs and BFs, $1 \le i \le p$} \\
    {- $r_m$:} & \multicolumn{2}{l}{ Minimum ring size, with $r_m \ge 3$} \\
    {- $s_t$:} & \multicolumn{2}{l}{ Minimum similarity threshold to classify record sets} \\
    \multicolumn{3}{l}{\textbf{Output:}} \\
    {- $\mathbf{M}$:} & \multicolumn{2}{l}{ Matching record sets} \\[1mm]
    1:  & $rings = group([P_1,P_2,\cdots, P_p],r_m)$&//group parties\\
    2:  & \textbf{for} $ring \in rings$ \textbf{do}:&//phase 1 \\
    3:  & ~~ $\mathbf{C}_{ring} = \{[\,]\}$; $\mathbf{M}_{ring} = \{[\,]\}$; $r = len(ring)$ & ~ \\
    4:  & ~~ \textbf{for} $i \in ring$ \textbf{do}:&//iterate parties\\    
    5:  & ~~ ~~ $records = \mathbf{D}^M_i.values()$ & ~ \\
    6:  & ~~ ~~ \textbf{for} $rec\_set \in \mathbf{C}_{ring}$ \textbf{do}: & ~ \\
    7:  & ~~ ~~ ~~ \textbf{for} $rec \in records$ \textbf{do}: & ~ \\
    8:  & ~~ ~~ ~~ ~~ $rec\_set.append(rec) $&//$\mathbf{C}$ in the ring\\
    9:  & ~~ $\mathbf{C'}_{ring} = sec\_sum(\mathbf{C}_{ring},ring)$&//summation\\
    10:  & ~~ \textbf{for} $cs \in \mathbf{C'}_{ring}$ \textbf{do}: & ~ \\
    11:  & ~~ ~~ $c = \mathbf{C'}_{ring}[cs]$ & ~ \\
    12:  & ~~ ~~ $Dice\_sim(cs) = \frac{r \times |\{\beta: c[\beta] = r, 1 \le \beta \le l\}|}{\sum_{\beta=1}^{l} c[\beta]}$&//Eq.~(\ref{eq:Dice_coefficient_cbf})\\ 
    13:  & ~~ ~~ \textbf{if} $Dice\_sim(cs) \ge s_t$ \textbf{then}:&//a match\\
    14:  & ~~ ~~ ~~ $\mathbf{M}_{ring}.append(cs)$&//$\mathbf{M}$ in the ring\\
    15:  & $\mathbf{C} = \{[\,]\}$; $\mathbf{M} = \{[\,]\}$&//initialize\\ 
    16:  & \textbf{for} $ring \in rings$ \textbf{do}:&//phase 2 \\
    17:  & ~~ $matches = \mathbf{M}_{ring}.values()$ & ~ \\
    18:  & ~~ \textbf{for} $match\_set \in \mathbf{C}$ \textbf{do}: & ~ \\
    19:  & ~~ ~~ \textbf{for} $match \in matches$ \textbf{do}: & ~ \\
    20:  & ~~ ~~ ~~ $match\_set.append(match)$&//$\mathbf{C}$ in all rings\\
    21:  & $\mathbf{C'} = sec\_sum(\mathbf{C},[P_1,P_2,\cdots,P_p])$&//summation\\
    22:  & \textbf{for} $cs \in \mathbf{C'}$ \textbf{do}: & ~ \\
    23:  & ~~ $c = \mathbf{C'}[cs]$ & ~\\
    24:  & ~~ $Dice\_sim(cs) = \frac{p \times |\{\beta: c[\beta] = p, 1 \le \beta \le l\}|}{\sum_{\beta=1}^{l} c[\beta]}$&//Eq.~(\ref{eq:Dice_coefficient_cbf})\\ 
    25:  & ~~ \textbf{if} $Dice\_sim(cs) \ge s_t$ \textbf{then}:&//a match\\
    26:  & ~~ ~~ $\mathbf{M}.append(cs)$&//$\mathbf{M}$ in all rings\\ [1mm]
      \hline \\ [2mm]
  \end{tabular*}
\end{table}

\subsection{Ring by ring communication}
\label{subsec-rbr_comm}

In the absence of a trusted $LU$, as is required by the previously described \emph{SEQ}
communication approach,
we propose a ring by ring communication pattern (\emph{RBR})
for comparing CBFs from multiple parties without using a $LU$.
This method is illustrated in Figure~\ref{fig:cbf_rbr} for three rings with $r=3$
parties in each ring, while Algorithm~4 outlines the steps of
\emph{RBR} in detail. 
Similar to \emph{SEQ}, parties are grouped into rings 
using the $group(\cdot)$ function (line~1 in Algorithm~4).
The value for $r_m$ in \emph{RBR} should be
set to $r_m \ge 3$, as a minimum of three parties are required in each ring
to perform secure summation without a $LU$.

The \emph{RBR} method consists of two phases. 
In the first phase (lines~2-14 in Algorithm~4), the parties in each ring
(lines~2-4) individually perform
secure summation among them using $sec\_sum(\cdot)$ (line~9) on their sets of candidate records $\mathbf{C}_{ring}$
(generated in lines~5-8) to generate the CBFs $\mathbf{C'}_{ring}$, and calculate their
similarities to identify matches $\mathbf{M}_{ring}$ in each ring (as shown in lines~10-14).
In the second phase in lines~16-26, all parties 
then perform secure summation among them on the
matches identified in each ring $\mathbf{M}_{ring}$ in order to identify matches
$\mathbf{M}$ from
all $p$ parties.

Every ring in the first phase can employ a different set of BF
parameters to reduce the possibility of collusion between 
a set of parties in different rings. In the second phase,
all parties then have to agree on another
set of parameters for BF encodings of the matches identified in rings
in the first phase.
In addition, the rings in the first phase can be processed independently
and in parallel in a distributed environment making the \emph{RBR} more scalable
(than the \emph{SEQ}) with larger dataset sizes. 

% --------------------------------------------------------------------

\section{Analysis of the protocol}
\label{sec-analysis}

In this section we analyze our multi-party PPRL protocol
in terms of complexity, privacy, and linkage quality.

\subsection{Complexity analysis}
\label{subsec_comp_analsis}

We assume $p$ parties participate in the protocol, each having a
database of $n$ records. We assume a private blocking/indexing
technique employed in the blocking step forms
$b \le n$ blocks for each party.
In Step~1 of our protocol, the agreement of parameters has a constant
communication complexity.  
Blocking the databases in Step~2 has $O(n)$
computation complexity at each party, and
finding the intersection of blocks from all parties has a
communication complexity of $O(p\cdot b)$ and a computation complexity of
$O(b\cdot log~b)$ at each party, as $p\cdot b$ BKVs need to be securely 
communicated, and for each of the $b$ BKVs a search operation of
$O(log~b)$ is required in order to identify the intersection set. 
Assuming the average
number of $q$-grams in the QID attributes $A$ of each record is $g$, 
the masking of QID values of records into
BFs of length $l$ using $k$ hash functions for $n$
records in Step~3 is $O(n\cdot g\cdot k)$ at each party.

Steps~4 and~5 consist of the secure summation of
BF vectors to calculate the 
CBFs of candidate sets. 
Our extended secure summation protocols \emph{HSS}
and \emph{SSS} aim to improve privacy at the cost of
complexity overhead. 
The extended \emph{HSS} protocol requires 
$n \cdot l \cdot p$ encrypted values (long integers of 4 bytes each) to be
exchanged among the parties, while the basic secure summation
(\emph{BSS}) and \emph{SSS} require exchanging
$n \cdot l \cdot p$
short integer values (of 2 bytes each), 
which is more efficient compared to \emph{HSS}.
Further, the 
homomorphic encryption and decryption functions used
in \emph{HSS} are computationally expensive compared
to
simple vector addition and subtraction operations~\cite{Lin09}.

With the simplified assumption that all blocks are of equal size
$n/b$, i.e. contain $n/b$ BFs at each party, 
then using the
\emph{NAI} communication method 
in each block $(n/b)^p$ candidate sets of BFs (i.e.\ all
candidate sets of records in a block) 
have to be generated and their
CBFs calculated. 
The first party performs $b(n/b)$ summations, the second party $b(n/b)^2$,
and finally $b(n/b)^p$ summations are performed by the
last party, leading to a total of $O(\sum_{i=1}^p b(n/b)^i)$
summations in Step~4. 

In Step~5, either the linkage unit $LU$ or 
the first party that initiated the secure summation protocol
performs a vector subtraction operation
(subtracts the random vectors from the summed vectors) on all the
candidate sets, resulting in $b(n/b)^p$ CBFs.
The similarities of these CBFs are then calculated
and matches classified, which requires $O(b(n/b)^p)$ computations.
This combinatorial complexity currently limits the
\emph{NAI} linkage
to a small number of parties or a large number of small
blocks (i.e.\ $p$ or $n/b$ has to be small). 

The communication patterns \emph{SEQ} and \emph{RBR} proposed in 
Sections~\ref{subsec-seq_comm} and \ref{subsec-rbr_comm}, respectively, 
improve Steps~4 and~5 significantly (depending on the 
number of parties per ring, $r$) by reducing the computation 
and communication complexities.
In general the complexities are reduced from the
exponential growth with $p$ down to $r+1$ in \emph{SEQ} and $max(r,p/r)$ in \emph{RBR}.
With the simplified
assumption that each ring has $r$ parties, the computation and
communication complexities of the \emph{SEQ} and \emph{RBR} methods are
as follows.

In the \emph{SEQ} method, $b(n/b)^r$ candidate sets are processed in
the first ring and $b(n/b)^{r+1}$ in each of the remaining rings
(i.e. $b(n/b)$ matching
sets, in the worst case, from the previous ring are compared with
the $b(n/b)^r$ sets in each ring),
resulting in total computation and communication complexities of 
$O(\sum_{i=1}^r b(n/b)^{i}$ + $(\sum_{i=1}^{r+1} b(n/b)^{i} \cdot (p/r-1))$ (with $r < p$).
The
computation and communication complexities of the \emph{RBR} method is 
$O(\sum_{i=1}^r b(n/b)^{i} \cdot (p/r)$ + $\sum_{i=1}^{p/r} b(n/b)^{i})$, where the first phase requires
$b(n/b)^r$ total candidate sets to be processed
in each of the $p/r$ rings and the second phase compares
the $b(n/b)$ matching sets from each ring.

Overall, the complexity of the \emph{NAI} method is $O(b(n/b)^p)$,
the \emph{SEQ} method is $O(b(n/b)^{r+1} \cdot p/r)$, and the
\emph{RBR} method is $O(max(b(n/b)^r \cdot p/r,b(n/b)^{p/r})$.
This theoretical analysis shows that
the two proposed communication methods \emph{SEQ} and \emph{RBR} 
are computationally efficient
compared to the \emph{NAI} method. Depending on the values for
$p$ and $r$, the \emph{SEQ} and \emph{RBR} methods outperform
each other.

The memory size required for a CBF is $2^x$, which is $2^x \ge p$,
bits for every position in the CBF. If the length of CBF is $l$,
the total memory consumption is $l \times \lceil log_2(p) \rceil$. For $p$ BFs
the memory required is $1$ bit for every position in the BF, 
and therefore the total memory consumption of $1 \times l \times p$. 
CBF requires relatively more memory than using BFs when $p$ is small, 
however with increasing
number of $p$ in multi-party PPRL, a CBF requires significantly
lower memory compared to BFs.
For example, when $p=5$ and $l=1,000$ the respective
memory sizes of a CBF and corresponding BFs are 
$1000 \times \lceil log_2(5) \rceil = 3000$ bits
and $1 \times 1000 \times 5 = 5000$ bits, while for $p=10$ and $l=1,000$
they are $1000 \times \lceil log_2(10) \rceil = 4000$ bits
and $1 \times 1000 \times 10 = 10000$ bits, respectively.

\subsection{Privacy analysis}
\label{subsec_privacy_analsis}

As with most of the existing PPRL approaches,
we assume that all parties follow
the 
honest-but-curious 
(semi-honest) adversary model~\cite{Vat13}, 
where
the parties follow the protocol while being 
curious to learn the other parties' data 
by means of inference attacks on input data~\cite{Vat14}
or collusion~\cite{Vat13}.
To analyze the privacy against inference attacks, 
we discuss what the parties can learn through inference
attacks or collusion during the protocol.
We assume a trusted
$LU$ is available, which does not collude with any parties,
as is commonly used in many practical PPRL applications~\cite{Ran13}.

However, collusion among the database owners is a privacy risk
in the basic secure summation protocol where a set of parties
can collude to learn the BF of another party using their
received summation values. To overcome this problem, 
in Section~\ref{sec-ext-secsum} we have
proposed two extended secure summation protocols, \emph{HSS} and \emph{SSS}.
The \emph{HSS} protocol uses homomorphic encryptions for secure 
summation which makes the protocol more secure because without knowing the
private key (which is only known to the $LU$) identifying a party's
BF values by means of collusion will not be successful.
However, this protocol encrypts each integer value in a BF (in total
$l$ values for each BF)
into a hash key (long integers) and thus incurs 
a very large communication overhead making
the protocol not viable for linkage of multiple large databases.
The \emph{SSS} protocol similarly makes the protocol more secure 
by adding additional integer values as salting keys 
by each party individually (known only
to the $LU$) in the secure summation, such that without knowing
the salting key value collusion among parties to learn a party's BF 
is not possible. Compared to \emph{HSS}, the \emph{SSS} approach does not
incur any expensive communication overhead as the salting keys
are small integer values.

Communication occurs among the parties in Step~1 of our
protocol (as described in Section~\ref{sec:algorithm})
where they agree on parameter settings, and
in Steps~4 and~5
where they
participate in a secure summation protocol.
The agreement of parameter settings in Step~1 does not
reveal any sensitive information about the underlying data.
Secure summation involves partial (masked) data exchange
where the parties communicate the partial and full summations 
of their BFs (1) among them in Step~4
and (2) to the $LU$
or the first party that initiated the communication in Step~5, respectively,
to calculate the CBFs of the candidate sets and their similarities.

The $LU$ (in \emph{SEQ} or \emph{NAI}) or the first party 
in each ring (in \emph{RBR}) receives the CBFs
of candidate sets in each ring. Compared to calculating
similarities of sets using their BFs directly, using a CBF
makes the inference attack on individual BFs and thus
their $q$-grams (strings) mapped into them more difficult. 
An inference attack allows an adversary to map a list of known
values from a global dataset
(e.g.\ $q$-grams or attribute values from a public telephone directory) 
to the encoded values (BFs or CBF)
using background information (such as frequency)~\cite{Vat14,Kuz11}.
The only information that can be learned from such an inference attack
using a CBF $c$
of a set of $x$ BFs 
(summed over $x$ parties, where either $x=p$ in \emph{NAI} and \emph{RBR}
methods or $x = r$ in \emph{SEQ} and \emph{RBR} methods)
is if a bit position in $c$ is either $0$ or $x$
which means it is set to $0$ or $1$, respectively, in the BFs
from all $x$ parties.

\begin{prop}
The probability of identifying the original (unencoded)
values of $x$ ($x > 1$)
individual records $R_i$ (with $1 \le i \le x$) given a single CBF $c$
is smaller than the probability of
identifying the original (unencoded values) of $R_i$ given
$x$ individual BFs $b_i$, $1 \le i \le x$.
\begin{eqnarray}
\label{eq:pr_inf} 
\forall_{i=1}^{x} Pr(R_i|c) < Pr(R_i|b_i) \nonumber
\\
\end{eqnarray}
\end{prop}

\begin{proof}
Assume the number of original (unencoded) values that can be mapped to a masked BF pattern $b_i$ from an inference attack is $n_g$. $n_g = 1$ in the worst case, where a one-to-one mapping exists between the masked BF $b_i$ and the original unencoded value of $R_i$. The probability of identifying the original value given a BF in the worst case scenario is therefore $Pr(R_i|b_i) = 1/n_g = 1.0$.
However, a CBF represents $x$ BFs and thus at least (in the worst case) $x$ original (unencoded) values, which leads to a maximum of $Pr(R_i|c) = 1/x$ with $x>1$ (when $x=1$, $c \equiv b_1$). Hence, $\forall_{i=1}^{x} Pr(R_i|c) < Pr(R_i|b_i)$.
\end{proof}

We will empirically evaluate and compare the amount of privacy
provided by masking records into CBFs and BFs against an
inference attack in Section~\ref{sec-discussion}.
A larger number of parties in a ring $r$ 
(i.e.\ the larger the value for $x$) results in an
increase in the difficulty of an inference attack 
(a smaller probability of suspicion $1/x$)
by the adversary ($LU$ or the first party in each ring
for \emph{SEQ} and \emph{RBR}, respectively) at the cost of more
candidate set comparisons.

Further, using different hash encodings $h_1,\cdots,h_k$ by different
rings in our \emph{SEQ} and \emph{RBR} methods improves privacy 
compared to the \emph{NAI} method by reducing the possibilities of
collusions, as discussed in Section~\ref{sec-comm-patterns}.
Since the hash functions used by parties in ring~1, for example, are not
known to parties in other rings, a collusion between parties in other
rings and / or the $LU$ will not be successful in inferring the original
values of parties in ring~1. A careful grouping of parties is therefore
required to improve privacy in a multi-party setting (for example,
$LU$ randomly groups or changes the grouping 
for different blocks). More specifically,
when $p$ parties are involved in the linkage,
the maximum number of possible combinations for collusion in \emph{NAI}
method is 
$\sum_{i=1}^p (p-1)$, while in \emph{SEQ} and \emph{RBR} it is
$\sum_{i=1}^{p/r} \sum_{j=1}^{r} (r -1)$. For example, with $p=9$
parties the \emph{NAI} method has $72$ possibilities to collude
while grouping into equal sized rings ($r=3$) leads to only $18$ different
collusion possibilities.

The values for the number of hash functions used ($k$) and the length
of the BF ($l$) 
provide a trade-off between the linkage quality and
privacy~\cite{Sch11}. The higher the value for $k/l$, the higher the
privacy and the lower the quality of linkage, because the number of
$q$-grams mapped to a single bit 
(and therefore the number of resulting collisions) increases, which leads to lower
linkage quality but makes it more difficult for an adversary to learn
the possible $q$-gram combinations~\cite{Kuz11}. 
The CLK encoding method
(as discussed in Section~\ref{subsec-bf}) of hash-mapping several QID
values from each record into one compound BF~\cite{Sch11,Vat12}
makes it even more difficult for an adversary to learn individual
QID values that correspond to a revealed 
bit pattern in a BF.

\subsection{Linkage quality analysis}
\label{subsec_quality_analsis}

Our protocol supports approximate matching of QID values, in that
data errors and variations are taken into account depending on
the minimum similarity threshold $s_t$ used.
The quality of BF-based masking depends on the
BF parameterization~\cite{Sch11,Vat14b}.
For a given BF length, $l$, and the number of elements $g$
(e.g. $q$-grams) to be
inserted into the BF, the optimal number of hash
functions, $k$, that minimizes the false positive rate $f$
(of a collision of two different $q$-grams being mapped to the
same bit position), is
calculated as~\cite{Mit05}
$k = l/g~ln(2)$,
leading to a false positive rate of
$f =  (1/2^{ln(2)})^{l/g}$.

While $k$ and $l$ determine the computational aspects of BF masking,
linkage quality and privacy will be determined by the false positive
rate $f$. A higher value for $f$ will mean a larger number of false
matches and thus lower linkage quality~\cite{Mit05,Sch11}. 
In our experimental evaluation we will set the BF
parameters for our approach according to the discussion presented here
and following earlier BF work in PPRL~\cite{Dur13,Sch11,Vat12,Vat14b}.

% --------------------------------------------------------------------

\section{Experimental evaluation}
\label{sec-experiment}

In this section, we empirically evaluate 
the performance of
our protocol (which we refer as \emph{AM-CBF} for
Approximate Matching with Counting BFs) 
with the \emph{SEQ}, \emph{RBR}, and \emph{NAI}
communication patterns
in terms of the three properties of PPRL,
scalability (complexity), linkage quality, and
privacy. 
We describe the competing baseline methods 
in Section~\ref{sec-competing-methods},
the datasets used in Section~\ref{sec-datasets},
the evaluation measures in Section~\ref{sec-measures},
and the experimental settings in Section~\ref{sec-settings}.
We then discuss the experimental
results in Section~\ref{sec-discussion}.

\subsection{Baseline methods}
\label{sec-competing-methods}

We use Lai et al.~\cite{Lai06}'s exact matching 
BF-based PPRL
approach (referred as \emph{EM-BF} for Exact Matching
with BFs) and Vatsalan and Christen~\cite{Vat14c,Vat16c}'s
approximate matching BF-based PPRL approach 
(\emph{AM-BF} for Approximate Matching with BFs) as competing
baseline methods to compare with our proposed approach.
Since other existing approaches for multi-party PPRL
(as reviewed in Section~\ref{sec-related})
are either based on expensive cryptographic techniques
or applicable to
categorical data only, we do not compare them with
our approach.

Lai et al.'s \emph{EM-BF} approach~\cite{Lai06} performs exact matching
of QIDs across multiple parties using BFs.
In their approach, the QID values of all records in a dataset
are first converted into one BF. 
Each party then partitions its BF into segments
according to the number of parties involved in the linkage, and sends
these segments to the corresponding other parties. The segments received by a party
are combined using a conjunction (logical AND) operation. The
resulting conjuncted BF segments are then exchanged between
the parties to construct the full conjuncted BF. 
Each party checks its own full BF of each record with the
conjuncted BF, and if the membership test is successful then the record
is considered to be a match. Though the cost of this approach is low
since the computation is completely distributed between the parties
and the processing of BFs is fast, the
approach can only perform exact matching. 

Vatsalan and Christen proposed \emph{AM-BF}~\cite{Vat14c,Vat16c} by
adapting the idea used in \emph{EM-BF}
of distributively computing
the conjunction of a set of BFs from multiple parties
to perform privacy-preserving approximate matching for multi-party PPRL.
Once the conjuncted BF segments are computed by the respective
parties, a secure summation protocol is initiated among the parties
to securely sum the number of common $1$-bits in the conjuncted 
BF segments as well as the total number of $1$-bits
in each party's BF. These two sums are then used to calculate
the Dice coefficient similarity of the set of BFs.
A filtering approach is employed to reduce the number of comparisons
based on segment similarity, such that if a sub-set of BF
segments of a candidate set (as calculated by a respective party)
has lower similarity then the 
BFs do not have to be compared with any of the BFs
from the other parties.

Both \emph{EM-BF} and \emph{AM-BF} approaches use the
\emph{NAI} method for comparing and classifying candidate sets
of records.

\subsection{Datasets}
\label{sec-datasets}

To provide a realistic evaluation of our approach, we based all our
experiments on a large real-world database, the North Carolina Voter
Registration (NCVR) database as available from
\texttt{ftp://alt.ncsbe.gov/data/}. This database has been used for the
evaluation of various other PPRL approaches~\cite{Dur13,Ran14,Ran15,Vat14c,Vat16c,Vat14}. We have
downloaded this database every second month since October 2011 and
built a combined temporal dataset that contains over 8 million records
of voters' names and addresses~\cite{Chr13NC}.
We are not aware of any available real-world dataset that contains
records from more than two parties that would allow us to evaluate 
multi-party PPRL approaches. We therefore generated, based on the real NCVR
database, a series of sub-sets for multiple parties, 
as will be described next.

To allow the evaluation of our approach with different number of
parties, with different dataset sizes, and with data of different
quality, we used and modified a recently proposed data
corruptor~\cite{Chr13b,Tra13} to generate various datasets with different
characteristics based on randomly selected records 
from the NCVR
database. 
The identifiers of the
selected and modified records were kept the same, 
which allows us to identify true and
false matches and therefore evaluate linkage quality, 
as described below.
Specifically, we extracted sub-sets of $5,000$, $10,000$,
$50,000$, $100,000$, $500,000$, and $1,000,000$ records from the NCVR
to generate datasets for $3$,
$5$, $7$, and $10$
parties, where
the number of matching records is set to 50\% (i.e.\ half of 
selected records occur in the datasets of all parties).

To evaluate how the approaches work with `dirty' data, 
we created several series of datasets for each of the
datasets generated above, 
where we included a varying number
of corrupted records into the sets of overlapping records ($0\%$,
$20\%$, and $40\%$). 
We applied various corruption
functions on randomly selected attribute values, 
including character edit operations (insertions, deletions,
substitutions, and transpositions), and optical character recognition
and phonetic modifications based on look-up tables and corruption
rules~\cite{Chr13b}.
This means that a certain
percentage of records in the overlap were modified for randomly
selected parties, while the original values were kept for the
other parties. Therefore,
some of these records are exact duplicates across some parties in a
set, but are only approximately matching duplicates across the other
parties in the set. This simulates, for example, the situation where
three out of five hospitals have the correct and complete contact
details (like name and address) of a certain patient, while in the
fourth and fifth hospitals some of the details of the same patient are
different.

\subsection{Evaluation measures}
\label{sec-measures}

We evaluate the three properties of PPRL 
using the
following evaluation measures.
The complexity (scalability) of linkage is measured by \emph{runtime}, 
\emph{communication size}, and the
\emph{number of comparisons} required for the linkage. 
The quality of the achieved linkage is measured using the 
\emph{F-measure}, calculated on classified matches and non-matches,
that has widely been used in record linkage,
information retrieval and data mining~\cite{Chr12}.

In line with other work in PPRL~\cite{Ran14,Vat14c,Vat16c,Vat14}, we
evaluate privacy using disclosure risk (DR) measures based on the
probability of suspicion, i.e.\ the likelihood a masked (encoded) database
record in $\mathbf{D}^M$ can be matched with one or several (masked) record(s) in a
publicly available global database $\mathbf{G}$. 
The probability of suspicion for a masked value/record $R^M$, $Pr(R^M)$,
is calculated as $1/n_g$ where
$n_g$ is the number of possible matches in $\mathbf{G}^M$
to the masked value $R^M$.
We conducted a frequency linkage attack~\cite{Vat14} 
by mapping the exchanged bit information in the BFs
generated from $\mathbf{D}^M$
to the BFs generated from $\mathbf{G}^M$.
We used the worst case scenario where $\mathbf{G} \equiv \mathbf{D}$,
because when
$\mathbf{G} \equiv \mathbf{D}$ there will be a one-to-one exact matching
of a global value for each value in $\mathbf{D}$.
Based on such linkage attack, we
calculate the
following disclosure risk measures, as proposed by Vatsalan et al.~\cite{Vat14}.

\begin{itemize}
\item \emph{Mean disclosure risk} ($DR_{Mean}$): This calculates the average risk 
      ($\sum_i^{|\mathbf{D}^M|}$ $Pr(R_i^M)/|\mathbf{D}^M|$) of any sensitive value in $\mathbf{D}^M$ being re-identified.

\item \emph{Marketer disclosure risk} ($DR_{Mark}$): This is calculated as the proportion of masked records 
in $\mathbf{D}^M$ that
match to exactly one masked record in $\mathbf{G}^M$ ($|\{R^M \in \mathbf{D}^M: Pr(R^M) = 1.0\}| /|\mathbf{D}^M|$).
\end{itemize}

\subsection{Experimental settings}
\label{sec-settings}

We implemented both our proposed approach and the competing 
baseline approaches in Python 2.7.3,
and ran all experiments on a server with 
four 6-core 64-bit Intel Xeon
2.4 GHz CPUs, 128 GBytes of memory and running Ubuntu 14.04. The
programs and test datasets are available from the authors. 

Following the discussion in Section~\ref{subsec_quality_analsis}
and other work in PPRL~\cite{Dur13,Sch09,Vat12}, 
we set the parameters as
BF length $l=500$, the number of hash functions
$k=20$, the length of grams $q=2$, the minimum similarity
threshold $s_t=0.8$, and the number of parties $p=[3,5,7,10]$. 
For the \emph{SEQ} method the minimum number of parties per
ring was set to $r_m = 2$ and
for the \emph{RBR} method $r_m=3$.
The attributes first name, last name, city, and
zipcode were used as the QIDs in the linkage.
We applied a
Soundex-based phonetic blocking~\cite{Chr12} 
for the private blocking step in all approaches,
which results in a set of blocks on which we individually conduct 
private comparison and classification.

The experiments are two-fold.
In the first part we evaluate the scalability of our approach
with proposed communication patterns and compare with the \emph{NAI} method.
In the second part we compare the complexity, linkage quality, and privacy
of our approach (with the \emph{SEQ} approach as it gave the best results
in the first part) with baseline methods.
We used the first name and last name attributes as the blocking keys
in the second set of experiments and all four attributes 
in the first set of scalability experiments in order to
allow for comparative evaluation with the \emph{NAI} method
(as we were unable to run experiments for the \emph{NAI} method on larger datasets
with larger $p$ when only two blocking attributes were used).

\begin{table}[t!]
\centering
 \begin{tabular}{c|ccc} \hline
 \\[-2mm]
  ~ & ~Runtime (sec)~ & ~Size (MBytes)~ & ~F-measure~ \\ 
  \hline
  \\[-2mm]
  \emph{BSS} & $34.03$ & $47.11$ & $1.0$ \\
  \emph{HSS} & $60557.93$ & $1257.06$ & $1.0$ \\
  \emph{SSS} & $34.03$ & $47.15$ & $1.0$ \\
  \hline
  \end{tabular}
 \caption{Results of secure summation protocols ($p=3$, $n=5000$).}
\label{sec_sum_results}
\end{table}

\subsection{Discussion}
\label{sec-discussion}

Table~\ref{sec_sum_results} shows the runtime and memory size
required and the F-measure results
achieved with the three secure summation protocols.
As can be seen, the \emph{HSS} requires significantly higher
runtime and memory (which is not practical in real applications)
to improve privacy against collusion attacks on the \emph{BSS}
without compromising the F-measure results. 
However, the
\emph{SSS} approach
requires similar runtime and memory as the \emph{BSS}
for improving privacy against collusions with no loss
in linkage quality.
We therefore use \emph{SSS} in all following experiments.

\begin{figure*}[!th]
  \centering
 \includegraphics[width=0.32\textwidth]{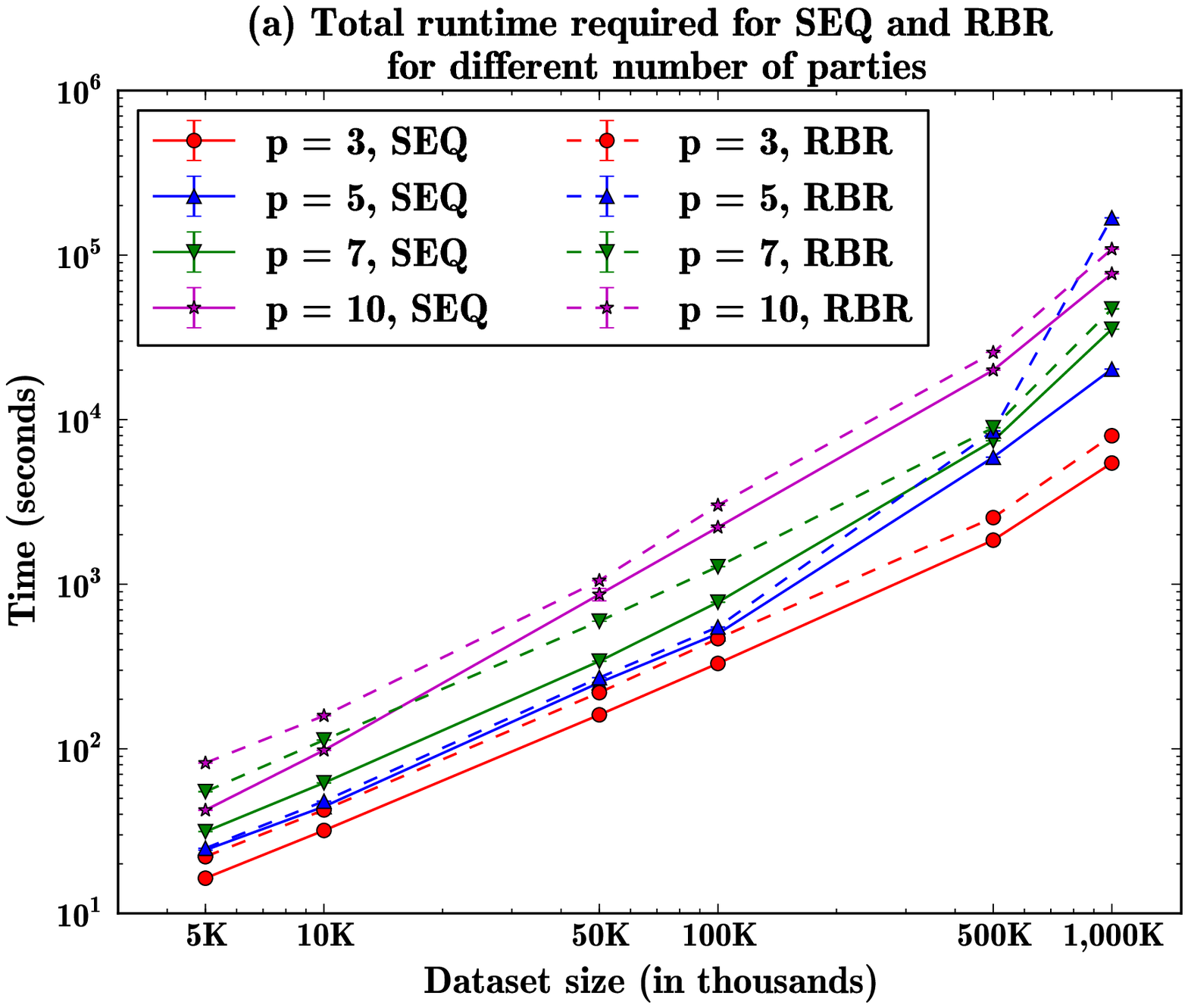}
 ~
 \includegraphics[width=0.32\textwidth]{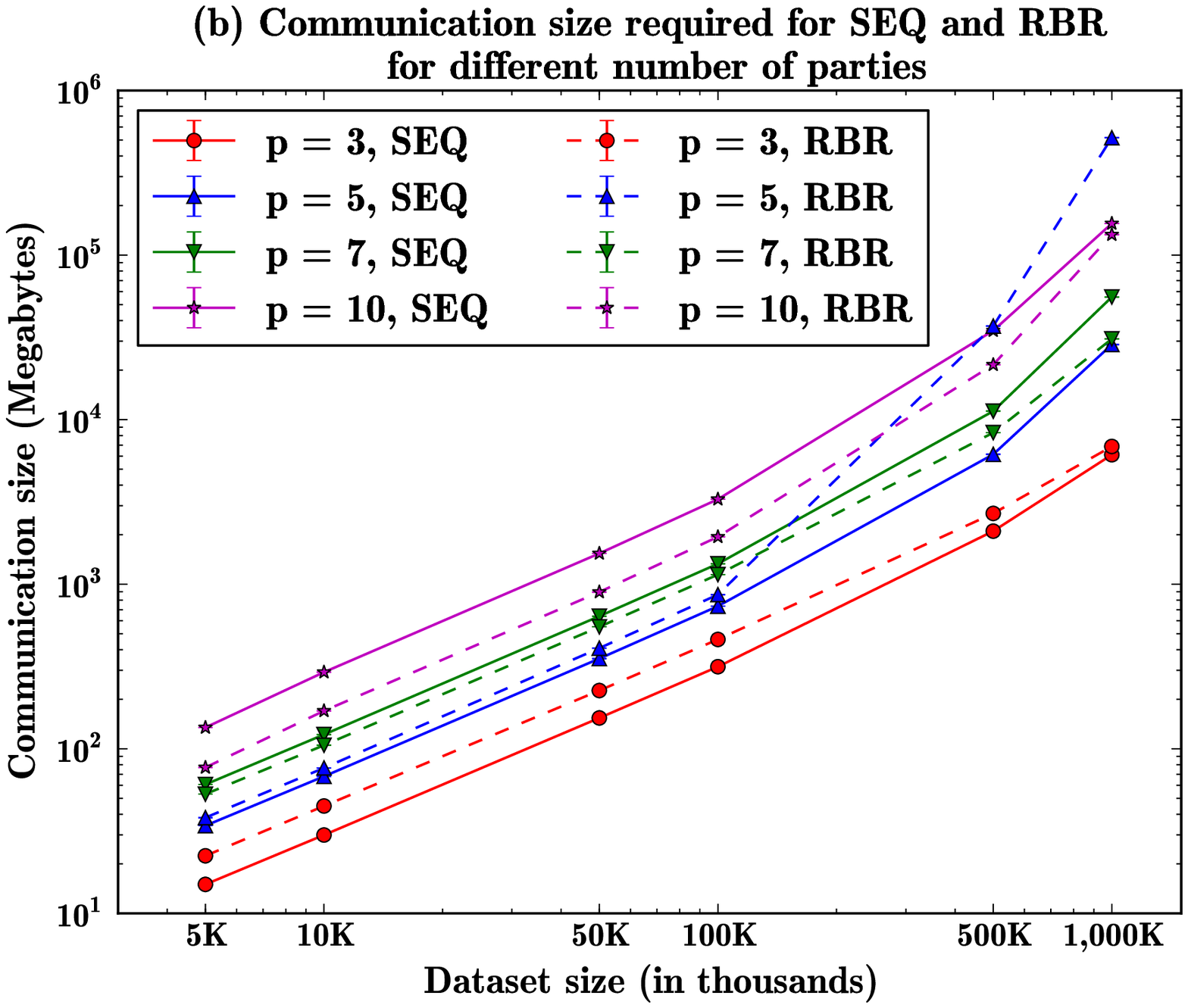}
~
 \includegraphics[width=0.32\textwidth]{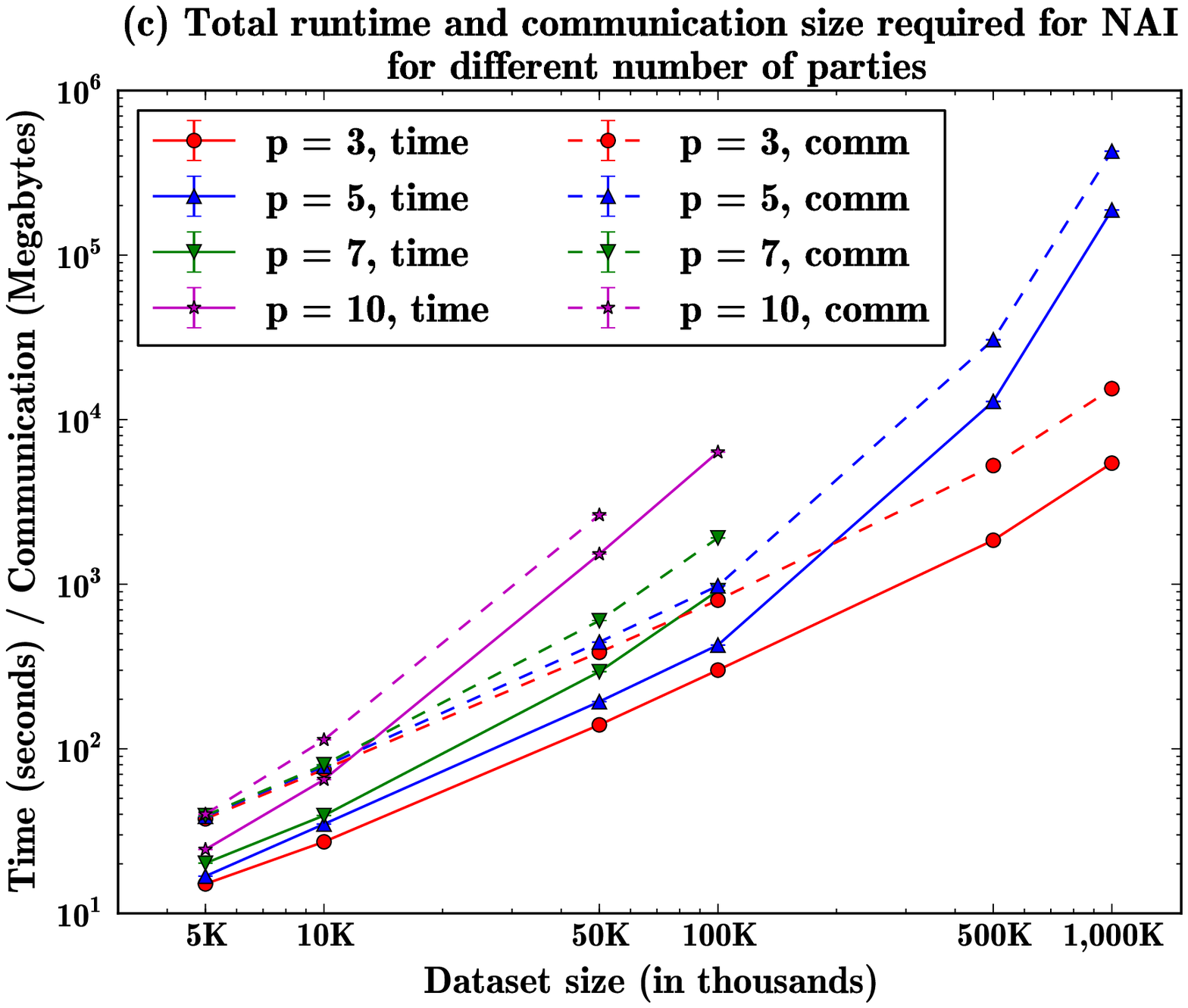}
~
 \includegraphics[width=0.32\textwidth]{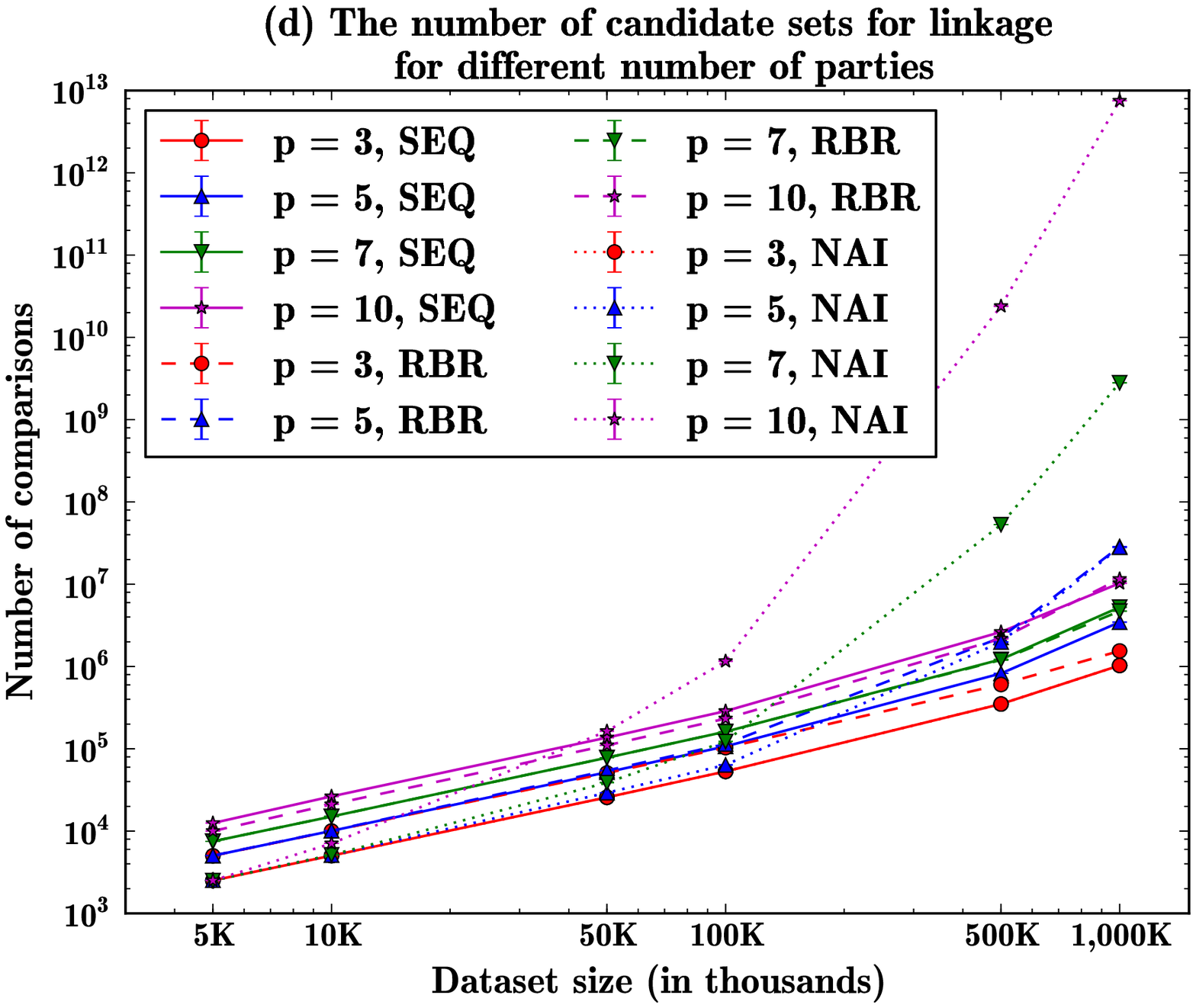}
~
 \includegraphics[width=0.32\textwidth]{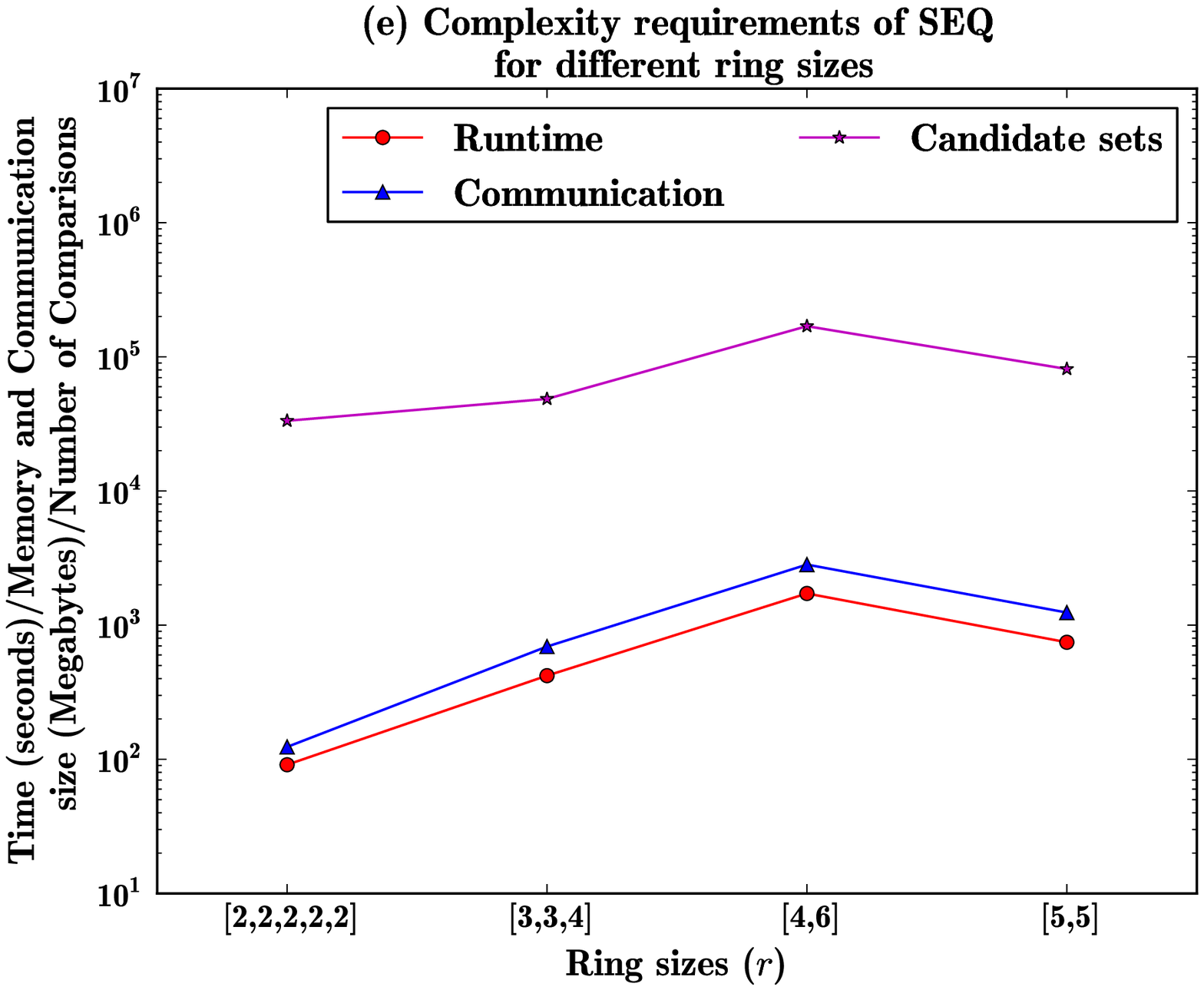}
~
 \includegraphics[width=0.32\textwidth]{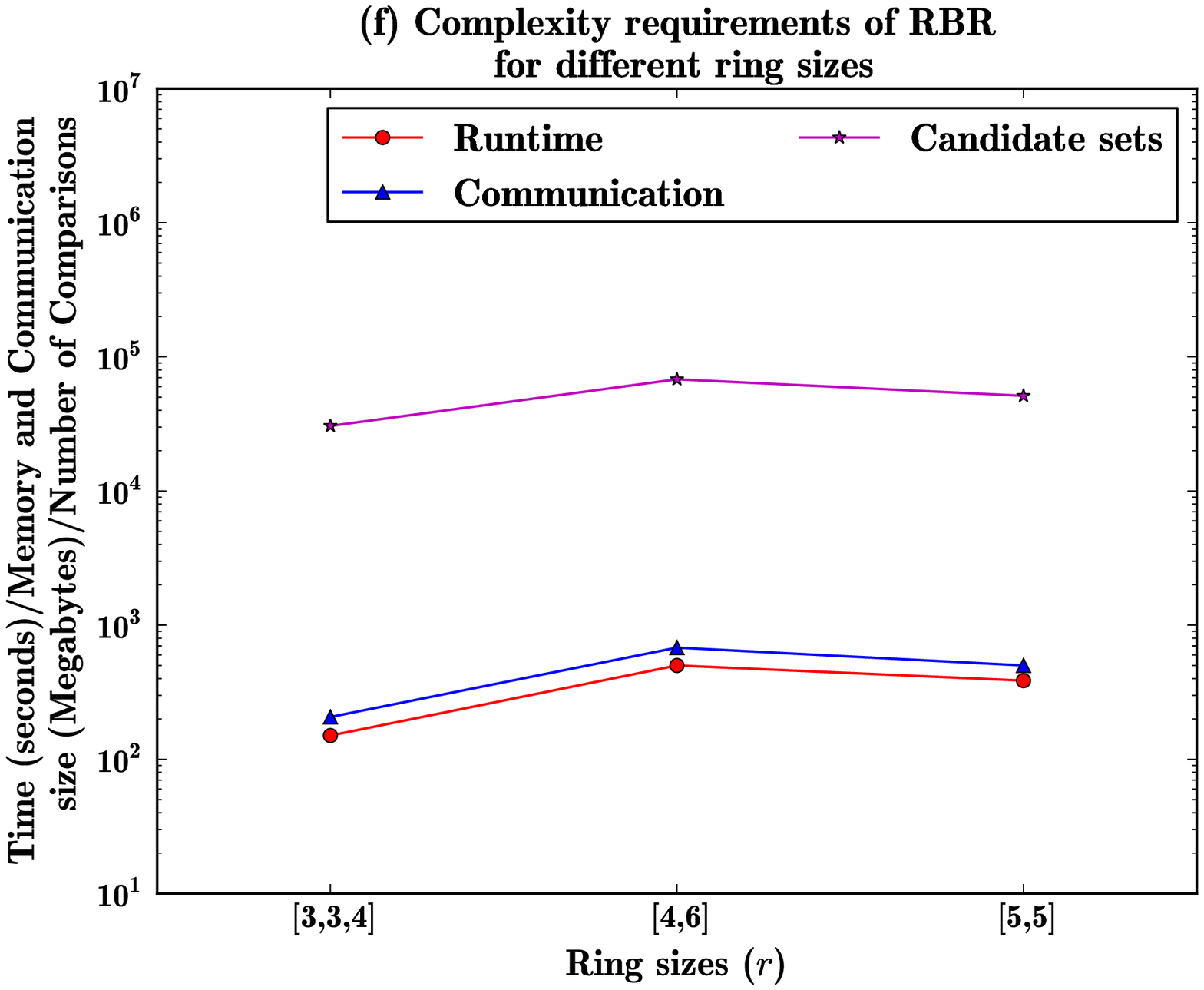}
 
  \caption{\small{(a) Total runtime and (b) communication size
    required for the \emph{SEQ} and \emph{RBR} and (c) \emph{NAI} methods,
    (d) the
    number of candidate sets to be compared for the linkage
    on different dataset sizes; and complexity requirements
    of (e) \emph{SEQ} and (f) \emph{RBR} for different ring sizes on the
    NCVR-10,000 datasets for $p=10$ parties. Note both axis of (a) - (d) 
    and the y-axis of (e) and (f) are in
    log scale. 
   }
    }
\label{fig:scal}
\end{figure*}

\begin{figure*}[!th]
  \centering
 \includegraphics[width=0.32\textwidth]{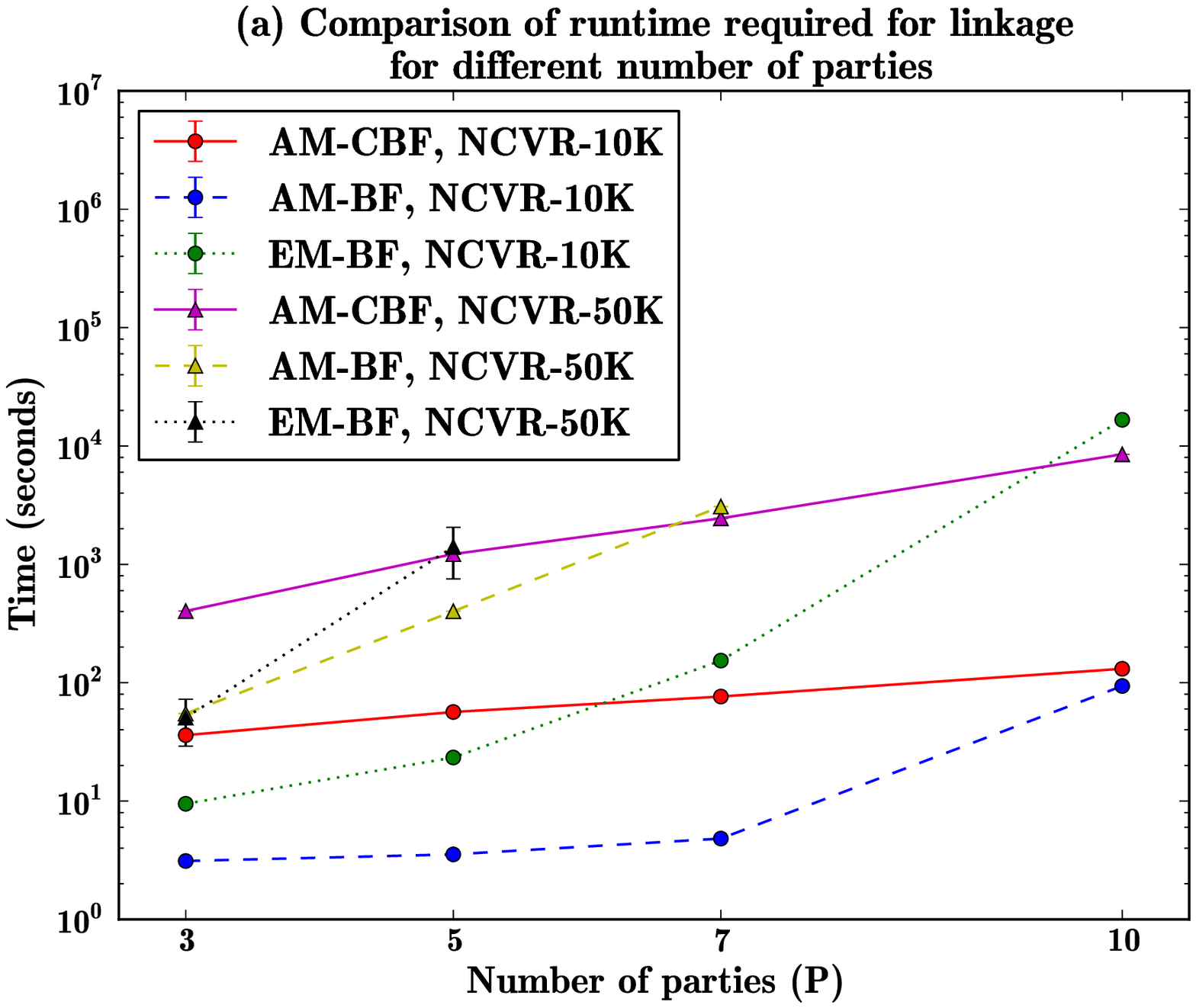}
~
 \includegraphics[width=0.32\textwidth]{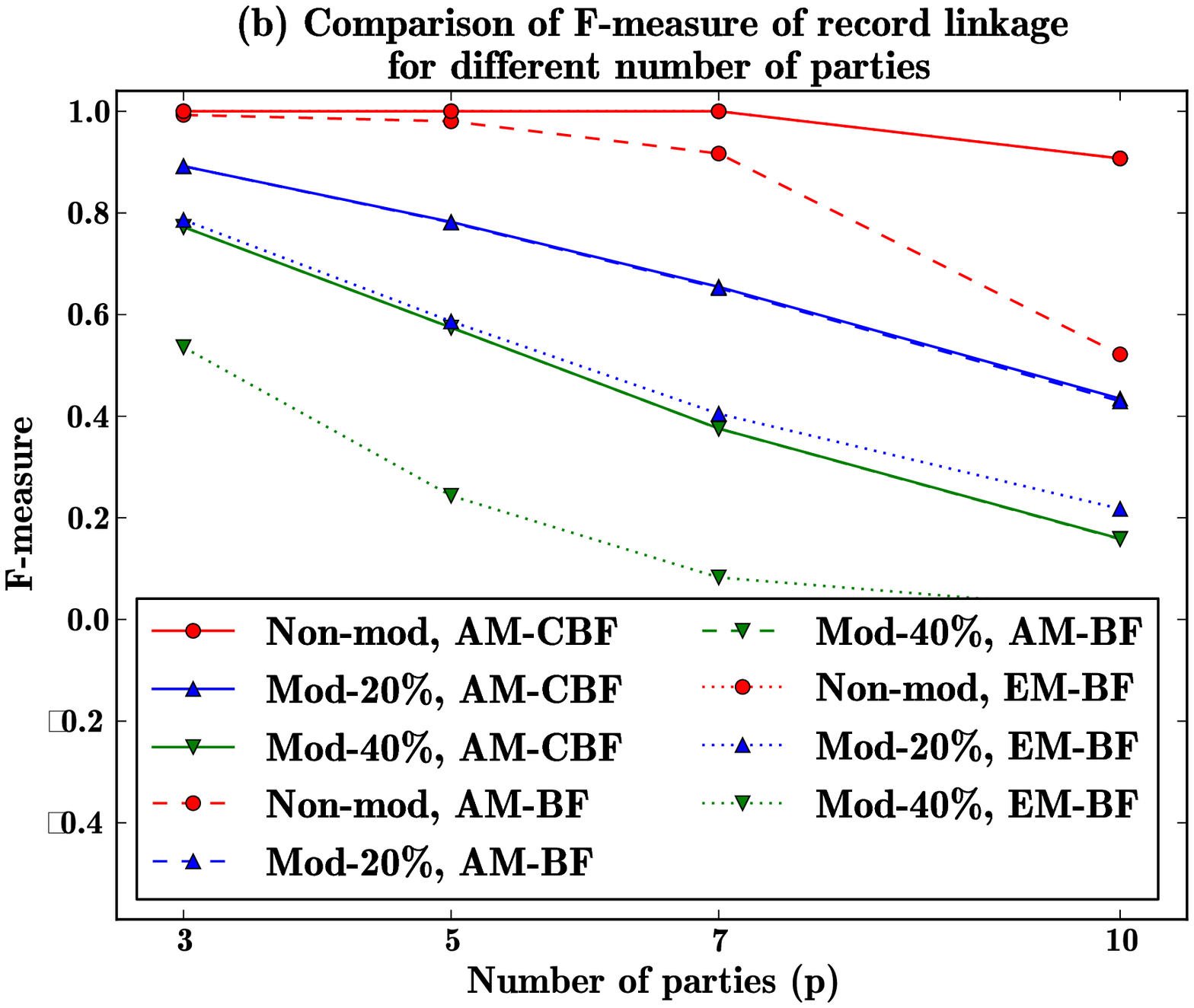}
~
 \includegraphics[width=0.32\textwidth]{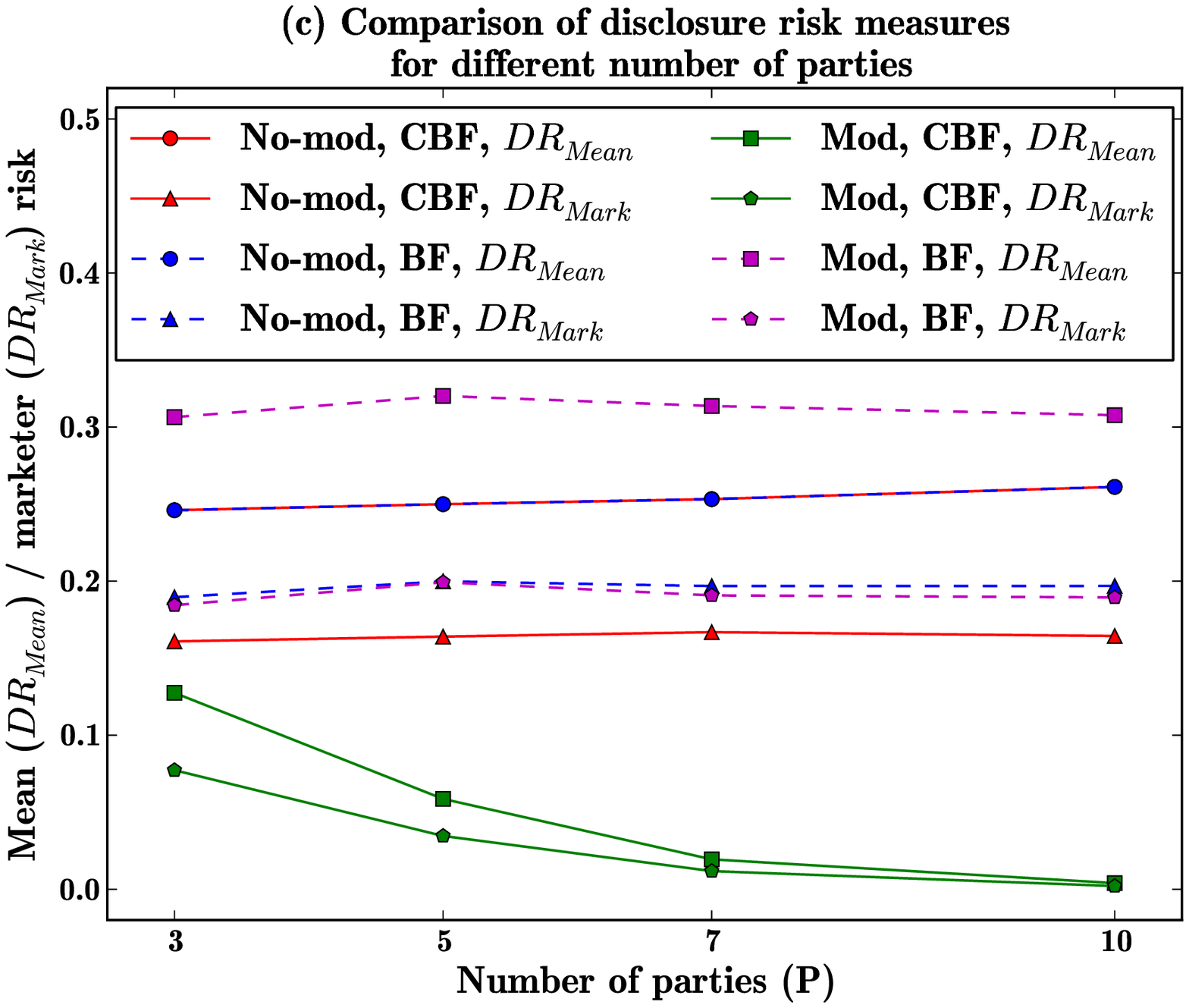}
 
  \caption{\small{Comparison of
  (a) runtime on the NCVR-10,000 and NCVR-50,000 datasets (missing data points are discussed in Section~\ref{sec-discussion}), 
  (b) F-measure of linkage, and (c) privacy results measured by mean disclosure risk ($DR_{Mean}$) and 
  marketer disclosure risk ($DR_{Mark}$) on the NCVR-10,000 datasets with competing baseline approaches.}
    }
\label{fig:comp}
\end{figure*}

Figures~\ref{fig:scal} (a) and (b)
show the complexity of our \emph{AM-CBF} approach
with the proposed communication patterns \emph{SEQ} and \emph{RBR}
in terms of runtime and communication size (comm), while
Figure~\ref{fig:scal} (c) shows the runtime and communication size
required by the \emph{NAI} communication pattern.
As the figures show the proposed communication methods improve the scalability
of our multi-party approach
by significantly reducing the number of candidate sets 
(and thereby the complexities)
compared to the \emph{NAI} method 
(which has been used in many existing solutions including
the baseline approaches).
As discussed in Section~\ref{subsec_comp_analsis},
\emph{SEQ} is faster than \emph{RBR} (Figure~\ref{fig:scal} (a)).
When $p=5$, \emph{RBR} is equivalent to \emph{NAI} (as we set $r_m=3$).
Both \emph{SEQ} and \emph{RBR}
are scalable (almost linear) with the
dataset size and the number of parties compared to the \emph{NAI}
method that increases
exponentially with the number of datasets and their sizes.
We were unable to conduct linkage experiments using the \emph{NAI} method for
larger $p$ and dataset sizes due to its exponential complexity
(and thus some data points for the \emph{NAI} method are missing
in Figure~\ref{fig:scal} (c)).

Figure~\ref{fig:scal} (d) shows the number of candidate
sets to be compared and classified using our \emph{AM-CBF} approach.
The number of candidate sets
with the \emph{NAI} method grows
exponentially with the size of the datasets for increasing $p$, as can be
seen when $p=10$.
In Figures~\ref{fig:scal} (e) and (f), we compare the 
complexity requirements of the \emph{SEQ} and \emph{RBR}
methods with different ring sizes based on $r_m$
on the NCVR-10,000 datasets for $p=10$ parties.
We evaluated the \emph{SEQ} method with
$r_m = [2,3,4,5]$ and the \emph{RBR} method
with $r_m=[3,4,5]$ ($r_m \ge 3$
in \emph{RBR}). 
This provides ring sizes of [2,2,2,2,2], [3,3,4], [4,6], and [5,5]
for $r_m=2$, $3$, $4$, and $5$, respectively.
As shown in the figures,
the complexity increases with larger ring sizes,
because larger ring sizes mean more comparisons are required
in each ring. 

We compare the runtime required by our approach (with \emph{SEQ})
with the baseline approaches 
on the NCVR-10,000 and 50,000 datasets in Figure~\ref{fig:comp} (a).
Though the baseline approaches 
require lower runtime for linking $p \le 5$ 
datasets than our approach (because the matches
identified in one ring need to be compared again
with the candidate sets in the other rings in our approach), 
they are not scalable to more parties ($p > 5$) and larger datasets.
We were unable to conduct experiments for \emph{AM-BF}
with $p > 7$ and \emph{EM-BF} with $p > 5$ on the NCVR-50,000
datasets due to the exponential increase in the number of candidate sets
generated.
As can be seen in Figure~\ref{fig:comp} (a), 
our approach is scalable and requires significantly lower runtime 
for linking $p=7$ and $10$ datasets compared to the baseline
approaches.

As shown in Figure~\ref{fig:comp} (b), our approach outperforms
the baseline \emph{EM-BF} and \emph{AM-BF}
approaches in terms of linkage quality measured by F-measure.
The filtering approach used in \emph{AM-BF} results in lower
F-measure compared to our approach.
The F-measure is high on the non-modified datasets (0\% corruption),
while it drops with $p$ on the modified (20\% and 40\%)
datasets as the number of missed true matches
increases when records are modified in each dataset.
However, our \emph{AM-CBF} approach achieves the highest 
F-measure on both modified and non-modified datasets
(though \emph{EM-BF} 
performs equally well 
on the non-modified
datasets).

Finally, we compared the DR measures
of privacy for our approach based on
CBF masking with 
baseline approaches based on
BF masking in Figures~\ref{fig:comp} (c) and (d).
The results show that the DR measures 
with our approach 
are consistently lower (and thus privacy
is higher) than the baseline
approaches. 
As expected, the DR measures also decrease with 
larger $p$ (as discussed in Section~\ref{subsec_privacy_analsis}),  
and with more corruptions in the dataset because corrupted
records reduce the probability of mapping to matching
global values in $\mathbf{G}$ to allow the records
to be re-identified.

% --------------------------------------------------------------------

\section{Related Work}
\label{sec-related}

Various techniques have been proposed in the literature tackling the
problem of PPRL~\cite{Vat13}. 
However, most of them consider linking two databases only, and
a only few approaches have been proposed for PPRL
on multiple databases.

An SMC-based approach using an oblivious
transfer protocol was proposed
by O'Keefe et al.~\cite{Kee04} 
for PPRL on multiple
databases. 
While provably secure, the approach only performs \emph{exact matching}
of masked values (i.e.\ variations and errors in the QIDs
are not considered) and it is \emph{computationally
expensive} compared to efficient perturbation-based privacy techniques,
such as BFs and $k$-anonymity~\cite{Vat13}. 
A multi-party $k$-anonymity-based PPRL approach
was introduced by
Kantarcioglu et al.~\cite{Kan08}.
In their apporach, a secure equi-join (\emph{exact matching}) 
is applied on the $k$-anonymized
databases by a $LU$ to identify matching records.

A multi-party PPRL approach for approximate matching 
of \emph{categorical values} 
was proposed by Mohammed et al.~\cite{Moh11}, where a top-down generalization is
performed to provide $k$-anonymous privacy
and the generalized blocks are then classified into matches and non-matches using the
C4.5 decision tree classifier.
Another efficient multi-party approach
for \emph{categorical data} was recently proposed
by Karapiperis et al.~\cite{Kara15}
using a Count-Min 
sketch data structure.
Sketches are used to summarize
the local set of elements which are then intersected to provide a 
global synopsis using homomorphic operations, secure summation, and symmetric noise
addition privacy techniques.

As described in Section~\ref{sec-competing-methods}, Lai et al.~\cite{Lai06}'s approach
uses efficient Bloom filter encoding for masking string data
in multi-party PPRL. However, the approach 
performs only \emph{exact matching}.
This approach has been adapted by Vatsalan and 
Christen~\cite{Vat14c,Vat16c} for \emph{approximate
matching} in multi-party PPRL (as described in detail in Section~\ref{sec-competing-methods}).
Several approximate comparison functions for calculating
similarities of \emph{pairs} of string and other data types
have been proposed for PPRL 
by adapting existing comparison functions. 

A secure version of the
Levenshtein edit distance
was proposed by Karakasidis and Verykios~\cite{Kar11a} using CBFs, where
the elements of a string are checked against a CBF of another string
to count the number of edits, 
while the CBF also provides the length of the string masked into it.
However, developing privacy-preserving comparison functions for multiple 
(more than two) values has only recently been considered
by Vatsalan and Christen~\cite{Vat14c,Vat16c}
using the Dice coefficient similarity.

Scalability has been addressed by using private blocking functions and / or
filtering approaches~\cite{Vat13,Vat16d} as well as by parallelizing and / or distributing
the tasks among several nodes or parties~\cite{Kara13,Vat16d}. 
The scalability problem in multi-party PPRL
has only recently been focused on. Ranbaduge et al.~\cite{Ran15,Ran16} proposed 
a family of efficient tree-based private blocking techniques for multi-party PPRL
using BF encoding and secure summation. Several filtering approaches have
been used in PPRL of two sources, including 
length filtering in BFs~\cite{Vat12} and
PPJoin techniques~\cite{Seh15}. 
The recently proposed multi-party
BF-based 
PPRL approach~\cite{Vat14c,Vat16c} (described above)
employs a filtering approach based on the similarity of
segments of BFs.
However, the \emph{number of comparisons} required for multi-party linkage
\emph{remains very large} even with existing private blocking and
filtering approaches employed~\cite{Ran15,Vat14c}. 
Therefore, efficient multi-party filtering and communication patterns
need to be developed 
in order to make PPRL scalable and viable in practical applications.

Several query tree representations 
have been used for optimizing multi-way join
queries~\cite{Lu91,Sch90}. 
Schneider and DeWitt~\cite{Sch90} studied query processing plans with
different types of structures, left-deep, right-deep, and bushy.
Left-deep and right-deep trees use a base table as the inner
and outer operand, respectively, of each join in the plan, while
in bushy trees both inputs to a join may themselves result from joins.
However, these techniques have not been used
for efficient processing of multi-party PPRL.

% --------------------------------------------------------------------

\section{Conclusion}
\label{sec:con}

We have presented an efficient and scalable
protocol for PPRL of multiple databases
with two improved communication patterns.
Our approach performs approximate matching
on the QID values (of string data)
masked using efficient and simple privacy techniques,
counting BFs and secure summation.
Experiments conducted on real datasets
showed the efficiency and scalability of
our approach compared to two baseline approaches
while achieving superior results 
in terms of linkage quality and privacy.

In future work, we plan to 
investigate other improved communication patterns, 
collision resistant secure summation protocols, 
and different approximate string comparison functions~\cite{Chr12}
to be incorporated.
We also aim to develop efficient PPRL techniques
for identifying matching record sets within sub-sets of parties,
which is an important research problem.

Another research direction would be to
develop multi-party PPRL approaches under
other adversary models such as the covert
model~\cite{Aum10} or accountable computing~\cite{Jia08} 
(where honest parties can verify fake data
from dishonest parties with high probability)
to overcome the limitations
of the semi-honest (HBC) adversary
model.
Finally, we plan to investigate improved
classification techniques for multi-party PPRL
including relational clustering and
graph-based approaches~\cite{Chr12} which are successfully used in
non-PPRL applications.

\section{Acknowledgments}
This work was partially funded by the Australian Research Council
(ARC) under Discovery Projects DP130101801 and DP160101934, and Universities
Australia and the German Academic Exchange Service (DAAD).

\bibliographystyle{abbrv}
\bibliography{paper}  

\end{document}